\documentclass[titlepage,runningheads]{llncs}

\usepackage[latin1]{inputenc} 
\usepackage{url} 
\usepackage{latexsym, stmaryrd}
\usepackage{amssymb, amsbsy, amsmath} 

\usepackage{comment}
\usepackage{ifthen}
\newboolean{extended-abstract}
\setboolean{extended-abstract}{false}
\ifthenelse{\boolean{extended-abstract}}%
{\excludecomment{RR}\includecomment{ABS}}%
{\includecomment{RR}\excludecomment{ABS}}
\usepackage{float}
\floatstyle{boxed}
\newfloat{mafigure}{thp}{pipo}
\floatname{mafigure}{Figure}

\newtheorem{fact}{Fact}{\bfseries}{\itshape}
\newcommand{\cqfd}{\hfill $\Box$}

\newcommand{\finex}{\hfill $\Diamond$}
\renewcommand{\paragraph}[1]{\vskip 1mm \noindent {#1}}

\newcommand{\Farc}[2]{\displaystyle \frac{#1}{#2}}
\newcommand{\F}{\mathcal{F}}  
\newcommand{\X}{\mathcal{X}}
\newcommand{\G}{\mathcal{G}}
\newcommand{\Pos}{\mathcal{P}\mathit{os}}
\newcommand{\rootp}{\mbox {\footnotesize $\Lambda$}}

\newcommand{\C}{\mathcal{C}}
\newcommand{\D}{\mathcal{D}}
\newcommand{\E}{\mathcal{E}}
\newcommand{\Hyp}{\mathcal{H}}
\newcommand{\R}{\mathcal{R}}
\newcommand{\RC}{{\R_\C}}
\newcommand{\RD}{{\R_\D}}
\newcommand{\T}{\mathcal{T}}
\newcommand{\sort}{\mathit{sort}}
\newcommand{\Ind}{\mathcal{I}}
\newcommand{\InferInd}{\vdash_{\Ind}}
\newcommand{\modelsind}{\models_{\mathcal{I}\mathit{nd}}}

\newcommand{\state}[1]{\sideset{_\llcorner}{_\lrcorner}{\mathop{#1}}}
\newcommand{\xstate}[2]{\sideset{_\llcorner}{}{\mathop{#1}}_{\raise 0.25ex \hbox{$\scriptstyle \lrcorner$}}^{#2}}
\newcommand{\dxstate}[2]{\mbox{$\xstate{#1}{#2}$}}

\newcommand{\lcons}{\llbracket}
\newcommand{\rcons}{\rrbracket}
\newcommand{\cons}[1][{\ }]{\,\mathopen{\lcons} {#1} \mathclose{\rcons}}

\newcommand{\mgi}{\operatorname{mgi}}
\newcommand{\depth}{\mathit{d}}
\newcommand{\NF}{\mathrm{NF}}

\newcommand{\sol}{\mathit{sol}}
\newcommand{\deriv}{\mathopen{:}}
\newcommand{\var}{\mathit{var}}

\newcommand{\mul}{\mathit{mul}}

\newcommand{\lpo}{\mathit{lpo}}

\newcommand{\Nat}{\mathsf{Nat}}

\newcommand{\Bool}{\mathsf{Bool}}
\newcommand{\List}{\mathsf{List}}

\newcommand{\Set}{\mathsf{Set}}

\newcommand{\true}{\mathit{true}}
\newcommand{\false}{\mathit{false}}

\newcommand{\sdec}{\mathit{dec}}

\newcommand{\senc}{\mathit{enc}}

\newcommand{\decorate}{\mathit{decorate}}
\newcommand{\ins}{\mathit{ins}}

\newcommand{\occur}{\in}
\newcommand{\occurbis}{\Subset}
\newcommand{\re}{\mathsf{Red}}
\newcommand{\sorted}{\mathit{sorted}}
\newcommand{\minl}{\mathit{min}}
\newcommand{\tie}{\mathit{tie}}
\newcommand{\rev}{\mathit{rev}}

\newcommand{\bad}{\mathsf{Bad}}

\def\frew#1#2#3#4#5#6#7#8{
\setbox0=\hbox{$#6 #7 #1 #8$}%
\setbox1=\hbox{$#6 #7 #2 #8$}%
\ifdim \wd0>\wd1 \rlap{\rlap{\hbox to \wd0{#5}}%
                            {\hbox to\wd0{\hfil\lower #3\box1\relax\hfil}}}{\raise #4\box0}%
\else \rlap{\rlap{\hbox to \wd1{#5}}{\hbox to\wd1{\hfil\raise #4\box0\relax\hfil}}}{\lower #3\box1}%
\fi
}

\def\fstep#1#2#3#4#5{\mathchoice{\frew{#1}{#2}{1.1ex}{1.1ex}{#5}{\scriptstyle}{#3}{#4}}%
                                {\frew{#1}{#2}{0.82ex}{1.1ex}{#5}{\scriptstyle}{#3}{#4}}%
                                {\frew{#1}{#2}{0.51ex}{0.82ex}{#5}{\scriptscriptstyle}{#3}{#4}}%
                                {\frew{#1}{#2}{0.51ex}{0.69ex}{#5}{\scriptscriptstyle}{#3}{#4}}}

\newcommand{\lrstep}[2]{\mathrel{\fstep{#1}{#2}{\;\>}{\>\>\;}{\rightarrowfill}}}
\newcommand{\rlstep}[2]{\mathrel{\fstep{#1}{#2}{\;\>\>}{\;\>}{\leftarrowfill}}}

\author{Adel Bouhoula\inst{1}
\and Florent Jacquemard\inst{2}}

\institute{%
Higher School of Communications of Tunis (Sup'Com),\\
University of November 7th at Carthage, 
Tunisia.
\email{adel.bouhoula@supcom.rnu.tn}
\and
INRIA Saclay-{\^ I}le-de-France and LSV, CNRS/ENS Cachan, France.
\email{florent.jacquemard@inria.fr}
}

\titlerunning{Automated Induction with Constrained Tree Automata}
\authorrunning{Adel Bouhoula and Florent Jacquemard}

\begin{document}

\begin{ABS}
\title{Automated Induction\\ with Constrained Tree Automata%
\thanks{A long version of this extended abstract
is available as a research report~\cite{BouhoulaJacquemard08rrlsv07}.}
\thanks{This work has been partially supported by INRIA/DGRSRT grants 06/I09 and 0804.}}
\end{ABS}

\begin{RR}
\title{Automated Induction\\ for Complex Data Structures 
\thanks{A preliminary version of these results appeared in 
the proceedings of the 4th {I}nternational {J}oint {C}onference on {A}utomated {R}easoning 
({IJCAR}'08) \cite{BouhoulaJacquemard08ijcar}.}
\thanks{This work has been partially supported by INRIA/DGRSRT grants 06/I09 and 08--04
and  a grant \emph{SSHN} of the French Institute for Cooperation in the French Embassy in Tunisia.}}
\end{RR}


\maketitle

\begin{abstract}
We propose a procedure for automated implicit inductive theorem proving for equational specifications made of rewrite rules with conditions and constraints.
The constraints are interpreted over constructor terms (representing data values), and may express syntactic equality, disequality, ordering and also membership in a fixed tree language.  
Constrained equational axioms between constructor terms are supported and can be used in order to specify complex data structures like sets, sorted lists, trees, powerlists...

Our procedure is based on tree grammars with constraints, a formalism which can describe exactly the initial model of the given specification (when it is sufficiently complete and terminating).
They are used in the inductive proofs first as an induction scheme for the generation of subgoals at induction steps, second for checking validity and redundancy criteria by reduction to an emptiness problem, and third for defining and solving membership constraints.

We show that the procedure is sound and refutationally complete.
It generalizes former test set induction techniques and yields natural proofs for several non-trivial examples presented in the paper, these examples are difficult 
to specify and carry on automatically with related induction procedures.

\end{abstract}


\section{Introduction}
Given a specification $\R$ of a program or system $S$
made of equational Horn clauses,
proving a property $P$ for $S$ generally
amounts to show the validity of $P$
in the minimal Herbrand model of $\R$,
also called \emph{initial model} of $\R$
(inductive validity).
In this perspective, it is important to have
automated induction theorem proving procedures
supporting a specification language
expressive enough to axiomatize
complex data structures 
like sets,  sorted lists, powerlists, 
complete binary trees, \textit{etc}.
Moreover, it is also important 
to be able to automatically generate induction schemas
used for inductive proofs
in order to minimize user interaction.
However, theories of complex data structures
generate complex induction schemes, and
the automation of inductive proofs is therefore difficult for such theories. 

It is common to assume that $\R$ is
built with \emph{constructor function symbols} 
(to construct terms representing data)
and 
\emph{defined symbols}
(representing the operations defined on constructor terms).
Assuming in addition the
\emph{sufficient completeness} of $\R$
(every ground (variable-free) term is reducible, 
using the axioms of $\R$, to a constructor term)
and the termination of $\R$, 
a set of representants 
for the initial model of $\R$
(the model in which we want to proof the validity of conjectures)
is the set of ground constructor terms not reducible by $\RC$
(the subset of equations of $\R$ between terms made of constructor symbols),
called constructor \emph{normal forms}.

In the case where the constructors are \emph{free}
($\RC = \emptyset$),
the set of constructor normal forms is simply the set of ground 
terms built with constructors and 
it is very easy in this case to define an induction schema.
This situation is therefore convenient for inductive reasoning,
and many inductive theorem provers require free constructors, 
termination and sufficient completeness.
However, it is not expressive enough to 
define complex data structures.
With rewrite rules between constructors, 
the definition of induction schema is more complex,
and requires a finite description of 
the set of constructor normal-forms.
%
Some progress has been done \emph{e.g.} in~\cite{BouhoulaJouannaud01}
and~\cite{BouhoulaJouannaudMeseguer} 
in the direction of handling specification with non-free constructors,
with severe restrictions (see related work below). 

Tree automata (TA) with constraints, or equivalently regular tree grammars with constraints,
have appeared to be a well suited framework
for the decision of problems related to term rewriting
(see~\cite{tata} for a survey).
This is the case for instance of ground reducibility, 
the property that all the ground instances of a given term are reducible
by a given term rewriting system (TRS).
This property was originally shown decidable for all TRS by David Plaisted~ \cite{plaisted85ic}; 
it is reducible to the (decidable) problem of emptiness for tree automata with disequality 
constraints (see e.g.~\cite{ComonJacquemard03}).
TA with constraints permit a finite representation of the set of constructor normal-forms
when $\RC$ is a left-linear TRS
(set of rewrite rules without multiple occurrences of variables in their left-hand-sides).
Indeed, on one hand TA can do linear pattern-matching,
hence they can recognize terms which are reducible by $\RC$,
and on the other hand, 
the class of TA languages is closed under complementation.
When the axioms of $\RC$ are not linear, or are constrained, 
some extensions of TA (or grammars) are necessary,
with transitions able to check constraints 
on the term in input, see e.g.~\cite{tata}.

\medskip
In this paper, we propose a framework for inductive theorem proving 
for theories containing
constrained rewrite rules between constructor terms and
conditional and constrained rewrite rules for defined functions.
The key idea is a strong and natural integration of tree grammars with constraints 
in an implicit induction procedure,
where they are used as induction schema.
%
%
Very roughly, our procedure starts with the
automatic computation of an induction schema,
in the form of a constrained tree grammar generating
constructor normal form.
This grammar is used later for the generation of subgoals from a conjecture $C$, 
by the instantiation of variables using the grammar's production rules,
triggering induction steps during the proof.
All generated subgoals are either deleted, following some criteria,
or they are reduced, using axioms or induction hypotheses,
or conjectures not yet proved, providing that they are smaller
than the goal to be proved. Reduced subgoals become then new conjectures
and $C$ becomes an induction hypothesis.
Moreover, constrained tree grammars are used as a decision procedure 
for checking the deletion criteria during induction steps.

%
%
Our method subsumes former test set induction procedures 
like~\cite{BouhoulaRusinowitch95jar,Bouhoula97jsc,BouhoulaJouannaud01},
by reusing former theoretical works on tree automata with constraints.
It is \emph{sound} and \emph{refutationally complete}
(any conjecture that is not valid in the initial model will be disproved)
when $\R$ is sufficiently complete  and the constructor subsystem $\RC$ is terminating.
Without the above hypotheses,
it still remains  sound and refutationally complete for
a restricted kind of conjectures, where all the variables are constrained to 
belong to the language of constructor normal forms. 
This restriction is expressible in the specification language (see below).
When the procedure fails, it implies that the conjecture
is not an inductive theorem, 
provided that $\R$ is \emph{strongly} complete
(a stronger condition for sufficient completeness) and ground confluent.
There is no requirement for \emph{termination} of the whole set of rules $\R$, 
unlike~\cite{BouhoulaRusinowitch95jar,Bouhoula97jsc},
but instead only for separate termination 
of the respective sets of rules for defined function and
for the constructors.

Moreover, if a conjecture $C$ restricted as above 
is proved in a sufficiently complete specification $\R$
and $\R$ is further consistently extended into $\R'$ with additional axioms 
for specifying \emph{partial} (non-constructor) functions,
then the former proof of $C$ remains valid in $\R'$,
see Section~\ref{sec:partial}.

The support of constraints permits in some cases to 
use the constrained completion technique
of~\cite{KirchnerKirchnerRusinowitch90ria}
in order to transform a non-terminating theory into a terminating one, 
by the addition of ordering constraints in constructor rules, 
see Section~\ref{se-ordering}.
It permits in particular to make proofs modulo non orientable axioms, without
having to modify the core of our procedure.
%

\medskip

We shall consider a specification of ordered lists as a running example 
throughout the paper.
Consider first non-stuttering lists
(lists which do not contain two equal successive elements)
built with the constructor symbols $\emptyset$ (empty list) 
and $\ins$ (list insertion) and following this rewrite rule:

\begin{equation}
\ins(x,\ins(x, y))  \to  \ins(x, y)
\tag{$\mathsf{c}_0$} \label{eq:ins-double}
\end{equation}

Rewrite rules can be enriched with constraints
built on predicates with a fixed interpretation on ground constructor terms.
For example, using ordering constraints built with $\succ$
we can specify ordered lists by the following axiom:

\begin{equation}
\ins(x_1,\ins(x_2, y))  \to  \ins(x_2, \ins(x_1, y))\cons[x_1 \succ x_2]
\tag{$\mathsf{c}_1$} \label{eq:ins-ord}
\end{equation}

Another interesting example is the case of membership constraints of the form $x : L$ where $L$ 
is a fixed regular tree language 
(containing only terms made of constructor symbols).
%
\begin{RR}
Such constraints can be useful in the context of system verification.
Assume that we have specified a defined symbol $\mathit{trace}$
characterizing the set of possible sequences of events of some system
i.e. $\mathit{trace}(\ell)$ reduces to $\mathit{true}$ iff
$\ell$ is a correct list of events (represented as constructor terms).
Now, assume also that we have defined a regular language $\mathit{Bad}$
(of ground constructor terms) representing lists of faulty events,
by mean e.g. of a (finite) tree grammar.
We can express in this way, for instance, that some undesirable event occurs eventually,
or that some event is always followed (eventually) by an expected answer,
or any kind of linear temporal property.
We can express with the constrained conjecture
\( \mathit{trace}(y) \neq \mathit{true} \cons[y : \mathit{Bad}] \)
that no bad list is a trace of the system.
Hence, showing that this conjecture is an inductive consequence
of the specification of the system amounts to do
verification of trace properties (i.e. reachability properties).
More details about this problematic, in the context of security protocol verification, are given in Section~\ref{sec:protocols}.

\end{RR}
We consider also stronger constraints 
which restrict constructor terms to be in normal form 
(i.e. not reducible by the axioms).
Let us come back to the example of non-stuttering sorted lists 
(sorted lists without duplication),
and add to the above rules the axioms below
which define a membership predicate $\occurbis$, 
using the information that lists are sorted:
\begin{align}
x \occurbis \emptyset  & \rightarrow \false
   \tag{$\mathsf{m}'_0$} \label{eq:occur2-empty}\\
x_1 \occurbis \ins(x_2, y_2)    & \rightarrow \true 
  {\cons[x_1 \approx x_2 ]}
  \tag{$\mathsf{m}'_1$} \label{eq:occur2-eq} \\
x_1 \occurbis y_1 & \rightarrow \false 
  {\cons[ y_1 \approx \ins(x_2, y_2), x_1 \prec x_2, y_1 \deriv \mathsf{NF}]}
  \tag{$\mathsf{m}'_2$} \label{eq:occur2-leq} \\
x_1 \occurbis  \ins(x_2, y_2) & \rightarrow x_1 \occurbis y_2 
  {\cons[x_2 \prec x_1]}
  \tag{$\mathsf{m}'_3$} \label{eq:occur2-geq} 
\end{align}

The constraint $y_1 \deriv \mathsf{NF}$ expresses the fact that this subterm is a
constructor term in normal form, i.e. that it is a sorted list.
Without this constraint, the specification would be inconsistent.
Indeed, let us consider the ground term 
$t = 0 \occurbis  \ins(s(0), \ins(0, \emptyset))$.
This term $t$  can be reduced into 
both $\true$ and $\false$, since $\ins(s(0), \ins(0, \emptyset))$ is not in normal form.
\begin{RR}
In Section \ref{sec:example}, we elaborate on these examples on sorted lists.
\end{RR}
%
Using constraints of the form $. : \mathsf{NF} {}$ as above
also permits the user to specify, 
directly in the rewrite rules, some ad-hoc reduction
strategies for the application of rewriting.
Such strategies include for instance several refinements of 
the innermost strategy
which corresponds to the {\em call by value} computation
in functional programming languages, where arguments are fully
evaluated before the function application.

\begin{RR}
Some non-trivial examples, including the above one, 
treated with our method are given in 
Section~\ref{sec:example} (sorted lists and verification of trace properties)
and Section~\ref{sec:powerlist} (powerlists).
Our procedure yields very natural and readable proofs on these examples
which are difficult (if not impossible) to specify and to carry on with 
the most of the other induction procedures.
\end{RR}


\paragraph{\bf Related work.}
The principle of our procedure is close to
test-set induction approaches~\cite{BouhoulaRusinowitch95jar,Bouhoula97jsc}.
The real novelty here is that test-sets are replaced by constrained tree grammars,
the latter being more precise induction schemes.
Indeed, 
they provide an \emph{exact} finite description of the initial model
of the given specification,
(under some assumptions like sufficient completeness 
and termination for axioms),
whereas cover-sets and test-sets are over-approximative in similar cases.

\begin{RR}
The soundness of cover-set~\cite{ZhangKKCADE88} 
and test-set~\cite{BouhoulaRusinowitch95jar,Bouhoula97jsc}
induction techniques 
do not require that the constructors are free. 
But, in this case, cover-sets and test-sets 
are over-approximating induction schemas,
in the sense that they may represent some reducible ground terms.
This may cause the failure (a result of the form ``don't know'') 
of the induction proof. 
On the other hand, the refutational completeness of test-set induction technique
is not guaranteed in this case.
\end{RR}

The first author and Jouannaud~\cite{BouhoulaJouannaud01} 
have used tree automata techniques to generalize test set induction 
to specifications with non-free constructors.
This work has been generalized in~\cite{BouhoulaJouannaudMeseguer} for membership equational logic.
These approaches, unlike the procedure presented in this paper, 
work by transforming the initial specification in order to get rid of 
rewrite rules for constructors.
Moreover, the axioms for constructors are assumed to be 
unconstrained and unconditional \emph{left-linear} rewrite rules,
which is still too restrictive for the specification of structures like sets or 
sorted lists...

\begin{ABS}
Kapur~\cite{Kapur} has proposed a method (implemented in the system \textsf{RRL}) 
for mechanizing cover set induction if the constructors are not free.
This handles in particular the specification of powerlists. 
We show in Section~\ref{se-powerlist} how our method can address similar problems.
\end{ABS}

\begin{ABS}
We describe in~\cite{BouhoulaJacquemard08rrlsv07} two proofs,
done resp. by Jared Davis and Sorin Stratulat,
of a conjecture on sorted lists,
done resp. by Jared Davis and Sorin Stratulat,
with \textsf{ACL2} using a library for ordered sets~\cite{Davis-osets}
and with 
\textsf{SPIKE}~\cite{BouhoulaRusinowitch95jar,Bouhoula97jsc,stratulat01}.
Both proofs require the 
addition of non-trivial lemmas 
whereas our procedure can prove the conjecture without additional lemma.
\end{ABS}

\begin{RR}
The theorem prover of \textsf{ACL2}~\cite{ACL2} 
is a new version of the Boyer-Moore theorem prover, \textsf{Nqthm}.
Its input language is a subset of the programming language \textsf{Common LISP}.
It is a very general formalism for the specification of systems,
and therefore permits in particular the specification of complex data structures mentioned above.
%
The example of sorted lists, presented in Section~\ref{sec:example}
can be processed with \textsf{ACL2}, but  the proof requires  the user
to add manually some lemmas, whereas the proof with our procedure does not require any lemma
(see Section~\ref{sec:acl2}).
The specification language of our approach is much less expressive than the one of \textsf{ACL2}, 
but the intention is to minimize the interaction with the user during the proof process,
in order to prevent the user from 
time consumption and the good level of expertise 
(both in the system to be verified and in the theorem prover)
which are often required in order to come up with the necessary key lemmas.
An interactive proof on the same specification with \textsf{SPIKE} 
is also presented in Section~\ref{sec:example}.

Kapur~\cite{Kapur} has proposed a method (implemented in the system \textsf{RRL}) 
for mechanizing cover set induction if the constructors are not free.
He defines particular specifications which may include in the declaration
of function symbols (including constructors) some \emph{applicability conditions}.
This handles in particular the specification of powerlists, 
as illustrated by some examples.
We show in Section~\ref{se-powerlist} how our method can address similar problems.

In~\cite{sengler96termination}, Sengler proposes a system \textsf{INKA} for automated
termination analysis of recursively defined algorithm 
over data types like sets and arrays.
It can handle constructor relations, under restrictions.
When it succeeds, this method provides an explicit induction scheme 
which can be exploited with an explicit inductive theorem proving procedure.

We lack a concrete base of comparison between 
our method and the  two above approaches,
because it was impossible for us to process our examples with
 \textsf{INKA}  (which is discontinued since 1997)  or  \textsf{RRL}. 
Let us outline some other important differences between our procedure
and these approaches.
The above explicit induction procedures 
are not well suited for the refutation of false conjectures.
When such a system fails, it is not possible to conclude
whether the conjecture is not valid or if the system need assistance from 
the user in order to complete the proof.
On the opposite, our implicit induction procedure 
is refutationally complete: any false conjecture will be refuted, 
under the assumptions mentioned above.
This property is of particular interest for debugging specifications 
of flawed systems or programs or also for the detection of attacks
on security protocols like in~\cite{BouhoulaJacquemard07arspa}
(see Section~\ref{sec:protocols}).
Finally, unlike explicit induction systems which are hierarchical, 
our procedure supports mutual induction. 
It is crucial for handling mutually recursive 
functions~\cite{Bouhoula97jsc}.
\end{RR}


\section{Preliminaries}
The reader is assumed familiar with the basic notions of term
rewriting~\cite{DershowitzJouannaud90} and first-order logic.  Notions and
notations not defined here are standard.

\paragraph{\bf Terms and substitutions.}
We assume given a many sorted signature $(\mathcal{S},\F)$ (or simply $\F$, for short) 
where $\mathcal{S}$ is a set of {\em sorts} and $\F$ is a finite set of function symbols with arities.
We assume moreover that the signature $\F$ comes in two parts, 
$\F = \C \uplus \D$ where $\C$ a set of \emph{constructor symbols},
and $\D$ is a set of \emph{defined symbols}.
Let $\X$ be a family of sorted variables. 
We sometimes denote variables with sort exponent like $x^{S}$ 
in order to indicate that $x$ has sort $S \in \mathcal{S}$.
The set of well-sorted terms over $\F$ 
(resp. constructor well-sorted terms) with variables in $\X$
will be denoted by $\T(\F, \X)$ (resp. $\T(\C, \X)$).  
The subset of $\T(\F, \X)$ (resp. $\T(\C, \X)$) 
of variable-free terms, or {\em ground} terms, 
is denoted $\T(\F)$ (resp. $\T(\C)$).
We assume that each sort contains a ground term.
The sort of a term $t\in \T(\F, \X)$ is denoted $\sort(t)$.

A term $t$ is identified as usual with a function from its
set of \emph{positions} (strings of positive integers) $\Pos(t)$ 
to symbols of $\F$ and $\X$, where positions are strings of positive integers. 
We denote the empty string (root position) by $\rootp$. 
The length of a position $p$ is denoted $|p|$.
The {\em depth} of a term $t$, denoted $\depth(t)$, 
is the maximum of $\{ |p| \mid p \in \Pos(t) \}$.
The {\it subterm} of $t$ at position $p$ is denoted by $t|_p$. 
The result of replacing $t|_p$ with $s$ at position $p$ in $t$ is denoted by $t[s]_p$. This notation 
is also used to indicate that $s$ is a subterm of $t$, in which case $p$ may be omitted.  
We denote the set of variables occurring in $t$ by $\var(t)$.
A term $t$ is {\em linear} if every variable of $\var(t)$ occurs exactly once in $t$.  

A \emph{substitution} 
is a finite mapping $\{ x_1 \mapsto t_1, \ldots, x_n \mapsto t_n \}$
where $x_1, \ldots, x_n \in \X$ and $t_1,\ldots t_n \in \T(\F, \X)$.
As usual, we identify substitutions with their morphism extension to terms.
A \emph{variable renaming} is a 
substitution mapping variables to variables.
We use 
postfix notation for substitutions application and composition. 
A substitution $\sigma$ is \emph{grounding} for a term $t$ 
if $t\sigma$ is ground.
\begin{RR}
The most general common instance of some terms $t_1, \ldots, t_n$ is denoted by 
$\mgi(t_1,\ldots,t_n)$.
\end{RR}

\paragraph{\bf Constraints and constrained terms.}
We assume given a constraint language $\mathcal{L}$, which is
a finite set of predicate symbols with a recursive Boolean interpretation
in the domain of ground constructor terms of $\T(\C)$.
Typically, $\mathcal{L}$ may contain the syntactic equality $. \approx .$
(syntactic disequality $. \not\approx .$),
some (recursive) simplification ordering $. \prec .$ on ground constructor terms
(for instance a lexicographic path ordering~\cite{DershowitzJouannaud90}),
and membership $.\,\deriv L$ to a fixed tree language $L \subseteq \T(\C)$
(like for instance the languages of well sorted terms or constructor terms in normal-form).  
%
\emph{Constraints} on the language $\mathcal{L}$ are Boolean combinations of atoms of the form
$P(t_1,\ldots,t_n)$  where $P \in \mathcal{L}$ and $t_1,\ldots,t_n \in \T(\C,\X)$.
By convention, an empty combination is interpreted to true.

The application of substitutions is extended from terms to constraints
in a straightforward way, and we may therefore 
define a solution for a constraint $c$ as a (constructor) 
substitution $\sigma$ grounding for all terms in $c$ 
and such that $c\sigma$ is interpreted to true.
The set of solutions of the constraint $c$ is denoted $\sol(c)$.
A constraint $c$ is \emph{satisfiable} if $\sol(c) \neq \emptyset$
(and \emph{unsatisfiable}  otherwise).

A \emph{constrained term} $t\cons[c]$ is a linear term $t \in\T(\F,\X)$ 
together with a constraint $c$, which may share some variables with $t$.  
Note that the assumption that $t$ is linear is not restrictive,  
since any non linearity may be expressed in the constraint, for instance $f(x, x)\cons[c]$ 
is semantically equivalent to $f(x, x') \cons[c \land x \approx x']$, where
the variable $x'$ does not occur in $c$.

\paragraph{\bf Constrained clauses.}
A \emph{literal} is an equation $s = t$ or a disequation $s \neq t$ or an oriented equation $s \to t$ between two terms.
A \emph{constrained clause} $C\cons[c]$ is a disjunction $C$ of literals
together with a constraint $c$.
A constrained clause $C\cons[c]$ is said to \emph{subsume} 
a constrained clause $C'\cons[c']$ if there is a substitution $\sigma$ such that $C\sigma$ is a sub-clause of $C'$ 
and $c' \land \lnot c\sigma$ is unsatisfiable. 

A \emph{tautology} is a constrained clause 
$s_1 = t_1 \vee \ldots \vee s_n = t_n \cons[d]$ 
such that  $d$ is a conjunction of equational constraints, 
$d = u_1 \approx v_1 \wedge \ldots \wedge u_k \approx v_k$
and 
there exists $i \in [1..n]$ such that $s_i \sigma = t_i \sigma$ where
$\sigma$ is the mgu of $d$.

\begin{RR}
\paragraph{\bf Orderings.}
%
A \emph{reduction ordering} is a well-founded ordering on $\T(\F, \X)$ 
monotonic wrt 
contexts and substitutions.
A \emph{simplification ordering} is
a reduction ordering 
which moreover contains the strict subterm ordering.
We assume from now on given a simplification ordering $>$
total on $\T(\F)$,
defined, \textit{e.g.}, on the top of a precedence as an lpo $\succ_\lpo$~\cite{DershowitzJouannaud90}. 
%


The \emph{multiset} extension $>^\mul$ of an ordering $>$
is defined as the smallest ordering relation on 
multisets such that $M \cup \{ t \} >^\mul M \cup \{s_1, \ldots, s_n \}$
if $t > s_i$ for all $i \in [1..n]$.
The extension $>_e$ of the ordering $>$ on terms 
to literals is defined as the multiset extension $>^\mul$ to 
the multisets containing the term arguments of the literals.
The extension of the ordering $>$ on terms to clauses is 
the multiset extension $>_e^\mul$ applied to the multiset of literals.
\end{RR}

\paragraph{\bf Constrained rewriting.}
A \emph{conditional constrained rewrite rule} is a constrained clause of the form
$\Gamma \Rightarrow l \to r \cons[c]$ such that
$\Gamma$ is a conjunction of equations, 
called the \emph{condition} of the rule,
the terms $l$ and $r$ (called resp. left- and right-hand side) are linear and 
  have the same sort, 
and $c$ is a constraint.
When the condition $\Gamma$ is empty, it is called a \emph{constrained rewrite rule}.
A set of conditional constrained, resp. constrained, rules is called a 
\emph{conditional constrained} (resp. \emph{constrained})  \emph{rewrite system}.
%
\begin{RR}

\end{RR}
Let $\R$ be a conditional constrained rewrite system.
The relation $s\cons[d]$ rewrites to $t\cons[d]$ by $\R$, 
denoted $s\cons[d] \lrstep{}{\R} t\cons[d]$,
is defined recursively by the existence of 
a rule $\rho \equiv \Gamma \Rightarrow \ell \to r \cons[c] \in \R$,
a position $p\in \Pos(s)$, 
and a substitution $\sigma$ such that 
$s|_p = \ell \sigma$, $t|_p = r\sigma$, 
$d\sigma \wedge \neg c \sigma$ is unsatisfiable,
and $u\sigma \downarrow_{\R} v\sigma$ for all $u=v \in \Gamma$.
%
The transitive and reflexive transitive closures, 
of $\lrstep{}{\R}$ are denoted 
$\lrstep{+}{\R}$ and $\lrstep{*}{\R}$, and $u \downarrow_{\R} v$
stands for
$\exists w,\ u \lrstep{*}{\R} w \rlstep{*}{\R} v$.

Note the semantical difference between conditions and constraints in  rewrite rules.
The validity of the condition is defined wrt the system $\R$
whereas the interpretation of constraint is fixed and independent from $\R$.

A constrained term $s\cons[c]$ is \emph{reducible} by $\R$ if there is some $t\cons[c]$ 
such that $s\cons[c] \lrstep{}{\R} t\cons[c]$.  
Otherwise $s\cons[c]$ is called \emph{irreducible}, or an $\R$-normal form.
A substitution $\sigma$ is {\em irreducible} by $\R$ if its image
contains only $\R$-normal forms.
A constrained term $t \cons[c]$ is {\em ground reducible} 
(resp. \emph{ground irreducible}) if $t\sigma$ is reducible (resp. irreducible)
for every irreducible solution $\sigma$ of $c$ grounding for $t$.

The system $\R$ is \emph{terminating} if there is no infinite
sequence $t_1 \lrstep{}{\R} t_2 \lrstep{}{\R} \ldots$,
$\R$ is \emph{ground confluent} if for any ground terms $u, v, w \in \T(\F)$, 
$v\rlstep{*}{\R} u \lrstep{*}{\R} w$, implies that $v \downarrow_{\R} w$, and
$\R$ is {\em ground convergent} if $\R$ is both ground confluent and terminating.
The {\em depth} of a non-empty set $\R$ of rules,
denoted $\depth(\R)$,
is the maximum of the depths of the left-hand sides of rules in $\R$.

\paragraph{\bf Constructor specifications.}
We assume from now on given a conditional constrained rewrite system $\R$.
The subset of 
$\R$ containing only function symbols from $\C$ is denoted $\RC$ 
and $\R \setminus \RC$ is denoted $\RD$.

%

\paragraph{\bf Inductive theorems.}
A clause $C$ is a \emph{deductive theorem} of $\R$ (denoted $\R \models C$)
if it is valid in any model of $\R$. 
A clause $C$ is an \emph{inductive theorem} of $\R$ 
(denoted $\R \modelsind C$)
iff for all for all substitution $\sigma$ grounding for $C$,
$\R \models C\sigma$.

We shall need below to generalize the definition of 
inductive theorems to constrained clauses as follows:
a constrained clause $C \cons[c]$ is an inductive theorem of $\R$ 
(denoted $\R \modelsind C \cons[c]$)
if for all substitutions $\sigma \in \sol(c)$ grounding for $C$
we have $\R \models C\sigma$.

\paragraph{\bf Completeness.}
A function symbol $f \in \D$ is \emph{sufficiently complete} wrt $\R$ 
iff for all $t_1, \ldots , t_n \in \T(\C)$, 
there exists $t$ in $\T(\C)$ such that $f(t_1, \ldots ,t_n)$ $\lrstep{+}{\R}$ $t$.  
We say that the system $\R$ is  sufficiently complete 
iff every defined operator $f \in \D$ is sufficiently complete wrt $\R$.
%
%
Let $f$ $\in \D$ be a function symbol and let:
\[ \Bigl\{\Gamma_1 \Rightarrow f(t^1_1,\ldots,t^1_k) \rightarrow r_1 \cons[c_1],\ldots,
\Gamma_n \Rightarrow f(t^n_1,\ldots,t^n_k) \rightarrow r_n \cons[c_n] \Bigr\}\]
be a maximal subset of rules of $\RD$ whose left-hand sides are identical 
up to variable renamings 
$\mu_1, \ldots, \mu_n$, \emph{i.e.}  
$f(t^1_1,\ldots,t^1_k)\mu_1 = f(t^2_1,\ldots,t^2_k)\mu_2 = \ldots
f(t^n_1,\ldots,t^n_k)\mu_n$.
We say that $f$ is \emph{strongly complete} wrt $\R$ (see~\cite{Bouhoula97jsc})
if $f$ is sufficiently complete wrt $\R$ and
$\R \modelsind \Gamma_1 \mu_1\cons[c_1 \mu_1] \vee \ldots \vee \Gamma_n
\mu_n\cons[c_n \mu_n]$ for every subset of $\R$ as above.
The system $\R$ is said strongly complete if every function symbol $f \in \D$ is
strongly complete wrt $\R$.



\begin{RR}

\section{Sorted Lists and Verification of Trace Properties} 
\label{se-example} \label{sec:example}
In this section, we present some examples for motivating the techniques introduced in this paper.
These examples illustrate the fact that our approach supports constraints in 
the axioms (both for constructor and defined functions) 
and the conjectures.
Note that constrained rules are not supported by test set induction procedures.

\subsection{Constructor Specification, Normal Form Grammar}
\label{sec:sorted-lists-grammar}
Consider a signature with sort 
\( \mathcal{S} = \{ \Bool, \Nat, \Set \} \), and constructor symbols:
\[  \C = \bigl\{ \true,\false:\Bool,\, 0:\Nat,\, s:\Nat \to \Nat,\, 
    \emptyset:\Set,\, \ins:\Nat \times \Set \to \Set \bigr\} 
\]
and a constructor rewrite system for ordered lists without duplication:
\[
\begin{array}{rcl}
\RC & = & 
\left\{ 
\begin{array}{rcl}
         \ins(x_1,\ins(x_2, y)) & \to & \ins(x_2, y)\cons[x_1 \approx x_2]\\
         \ins(x_1,\ins(x_2, y))  & \to  & \ins(x_2, \ins(x_1, y))\cons[x_1 \succ x_2]
\end{array}
\right\}
\end{array}
\]
Note the presence of constraints in these rewrite rules.
The equality constraint in the first rule permits the elimination of (successive) redundancies in lists,
and the ordering constraint in the second rule ensures that the application of this rule 
will sort the lists.
Note that the first rule actually corresponds to the unconstrained rewrite rule:
$ \ins(x, \ins(x, y)) \to \ins(x, y)$.
As outlined in introduction, this rule cannot be handled by the procedures 
of~\cite{BouhoulaJouannaud01,BouhoulaJouannaudMeseguer}, 
because it is not not left-linear.

Constrained grammar 
are presented formally in Section~\ref{se-grammar}.
In this section, we shall only give a taste of 
this formalism and how their are used in the automatic inductive proof of conjectures. 

\noindent The set of ground $\RC$-normal forms is described by the following
set of patterns:
\[ 
\begin{array}{lcl}
\NF(\RC) =  & & \{ x :\Bool \} \cup \{ x :\Nat \} \cup \{ \emptyset \}
   \cup \{ \ins(x,\emptyset)  \mid x:\Nat \}\\
 & \cup & \{ \ins(x_1, \ins(x_2, y)) \mid x_1, x_2:\Nat, \ins(x_2, y) \in \NF(\RC), x_1 \prec x_2 \}
\end{array} \] 
We build a constrained grammar $\G_\NF(\RC)$ which generates $\NF(\RC)$
by means of non-terminal replacement guided by some production rules.
The four first subsets of $\NF(\RC)$ are generated by a tree grammar
from the four non-terminals: 
$\bigl\{ \xstate{x}{\Bool}, \xstate{x}{\Nat},\xstate{x}{\Set},\state{\ins(x, y)} \bigr\}$
and using the production rules
(the non terminals are considered below modulo variable renaming):
\[ 
\begin{array}{rclcrclcrcl}
\xstate{x}{\Bool} & := &  \true & &
\xstate{x}{\Bool} & := & \false\\[1mm]
\xstate{x}{\Nat} & := & 0 & & 
\xstate{x}{\Nat} & := & s(\xstate{x_2}{\Nat})\\[1mm]
\xstate{x}{\Set} &  := & \emptyset & & 
\state{\ins(x, y)} & := & \multicolumn{5}{l}{\ins(\xstate{x}{\Nat},\xstate{x}{\Set})}
\end{array}
\]
For the last subset of $\NF(\RC)$, we need to apply the negation of the constraint 
$x_1 \approx x_2 \lor x_1 \succ x_2$ in the production rules of the grammar.
For this purpose, we add the production rule:
\[  
\state{\ins(x, y)} \quad := \quad 
\ins(\dxstate{x}{\Nat}, \state{\ins(x_2, y_2)})\cons[x^\Nat \prec x_2] 
\] 
Note that the variables in the non terminal $ \state{\ins(x_2, y_2)}$ in the right member
of the above production rule have been 
renamed in order to be distinguished from the variables in the non terminal in the left member.

\subsection{Defined Symbols and Conjectures}
We complete the above signature with the set of defined function symbols:
\[ 
\D = \{ \sorted: \Set \to \Bool, 
\occur, \occurbis: \Nat \times \Set \rightarrow \Bool \}
\]
and the conditional constrained TRS $\RD$ containing the following rules:
\begin{align}
\sorted(\emptyset) & \to \true
 \tag{$\mathsf{s}_0$} \label{eq:sorted-empty}\\
\sorted(\ins(x,\emptyset)) &\to \true
 \tag{$\mathsf{s}_1$} \label{eq:sorted-single} \\
\sorted(\ins(x_1,\ins(x_2, y))) & \to \sorted(\ins(x_2, y)) \cons[x_1 \prec x_2]
  \tag{$\mathsf{s}_2$} \label{eq:sorted-2} 
\end{align}
Note that there is no axiom for the case $\cons[x_1 \succeq x_2]$.
The defined function $\sorted$ is nevertheless sufficiently complete wrt $\R$.
We can show with an induction (on the size of the term)
that every term $t$  of the form  $\sorted(\ins(t_1,\ins(t_2, \ell)))$  
can be reduced to a constructor term.
If $t_1 \prec t_2$, then (\ref{eq:sorted-2}) applies and the term obtained is smaller than $t$.
If $t_1 \succeq t_2$, then $t$ is reducible by $\RC$ into the smaller
$\sorted(\ins(t_2, \ell)))$ if $t_1 \approx t_2$ or
into $\sorted(\ins(t_2,\ins(t_1, \ell)))$ if $t_1 \succ t_2$, 
and this latter term is furthermore reduced by the rule (\ref{eq:sorted-2}) of $\RD$ into 
$\sorted(\ins(t_1, \ell))$.

The rules (\ref{eq:occur2-empty}-\ref{eq:occur2-geq}) implements a membership test restricted to ordered lists.
The function $\occur$ specified below another variant of a membership test on lists.
\begin{align}
x \occur \emptyset  & \rightarrow \false
   \tag{$\mathsf{m}_0$} \label{eq:occur-empty}\\
x_1 \occur \ins(x_2, y)  & \rightarrow \true \cons[x_1 \approx x_2]
 \tag{$\mathsf{m}_1$} \label{eq:occur-eq} \\
x_1 \occur \ins(x_2, y)  & \rightarrow x_1 \occur y \cons[x_1 \not\approx x_2] 
  \tag{$\mathsf{m}_2$} \label{eq:occur-neq}
\end{align}
\begin{align}
x \occurbis \emptyset  & \rightarrow \false
   \tag{$\mathsf{m}'_0$}\\
x_1 \occurbis \ins(x_2, y_2)    & \rightarrow \true \cons[x_1 \approx x_2 ]
  \tag{$\mathsf{m}'_1$}\\
x_1 \occurbis y_1  & \rightarrow \false \cons[ x_1 \prec x_2, y_1 \deriv \state{\ins(x_2, y_2)}]
  \tag{$\mathsf{m}'_2$}\\
x_1 \occurbis  \ins(x_2, y_2) & \rightarrow x_1 \occurbis y_2 \cons[x_2 \prec x_1]
  \tag{$\mathsf{m}'_3$}
\end{align}
Like $\sorted$, the defined functions $\occur$ and $\occurbis$ are sufficiently complete wrt $\R$.

The above version of the rule (\ref{eq:occur2-leq}) is the formal one
(the version in introduction was given in a simplified notation).
Note the presence of the membership constraint 
$y_1 \deriv \state{\ins(x_2, y_2)}$ in~(\ref{eq:occur2-leq}).
It refers to the above normal form grammar $\G_\NF(\RC)$ and hence restricts
the variable $y_1$ to be a constructor term headed by $\ins$ and in normal form.

One may wonder why we added this membership constrained and 
why a rule ($\mathsf{m}''_2$) of the form 
$x_1 \occurbis \ins(x_2, y_2) \rightarrow \false \cons[x_1 \prec x_2]$
would not be satisfying.
The reason is that with the rule ($\mathsf{m}''_2$) instead of~(\ref{eq:occur2-leq}), 
the specification is not consistent.
Indeed, let us consider the ground term $t = 0 \occurbis  \ins(s(0), \ins(0, \emptyset))$.
Note that $t$ is not in normal form.
It can be rewritten on one hand into
$0 \occurbis \ins(0, \ins(s(0), \emptyset))$ by $\RC$,
which is in turn rewritten  into $\true$ using (\ref{eq:occur2-eq}).
On the other hand, $t$ can be rewritten into $\false$ by ($\mathsf{m}''_2$).
This second rewriting is not possible with (\ref{eq:occur2-eq}), because of 
the membership constraint in this rule.

Another idea to overcome this problem should be to add a condition as in:
\begin{align}
\sorted(y) = true \Rightarrow
x_1 \occurbis \ins(x_2, y)  & \rightarrow \false \cons[x_1 \prec x_2]
  \tag{$\mathsf{m}'''_2$} \label{eq:occur4-leq} 
\end{align}
The specification with~(\ref{eq:occur4-leq}) is inconsistent as well since
the term $t$ is rewritten by $\RC$ into
$\sorted(\ins(0, \ins(s(0), \emptyset))$, 
which is rewritten into $\true$ by $\RD$.
Therefore, the addition of the membership constraint in rule~(\ref{eq:occur2-leq})
is necessary for the specification of $\occurbis$.


\bigskip
\noindent Let us consider 
the two following conjectures that we are willing to prove by induction:
\begin{align} 
\sorted(y) & =  \true \label{ex-goal1} \\ 
x \occurbis y & =  x \occur y \label{ex-goal-occur} 
\end{align}

\subsection{Test Set Induction} \label{sec:ex-test-set}
Roughly, the principle of a proof by 
test set induction~\cite{BouhoulaRusinowitch95jar,Bouhoula97jsc}
is the one presented in introduction except that:
\begin{enumerate}
\item the induction scheme is a \emph{test set}  (a finite set of terms).
\item \label{it:test-set-instanciation} variables in the goals are instantiated by terms from the test set.
\end{enumerate}
Moreover, the instantiation in~\ref{it:test-set-instanciation}  can be 
restricted to so called \emph{induction variables} (see~\cite{Bouhoula97jsc}),
which are the variables occurring (in a term of a goal)
at a non-variable and non-root position
of some left-hand sides of rules of $\RD$. 

Let us try to prove (\ref{ex-goal1}) using the test set induction technique.
A test set\footnote{This test set is an over approximating description of the 
set of constructor terms in normal form. 
For instance, the term
$\ins(s(0), \ins(0, \emptyset))$ is an instance of the third element of the test set 
but it is not in normal form.}
for $\R$ (and sort $\Set$)  has to contain:
\[ \mathcal{TS}(\Set,\R) = \bigl\{\emptyset,\, \ins(x_1,\emptyset),\,\ins(x_1, \ins(x_2, y))\bigr\} \]

We start by replacing $y$ in (\ref{ex-goal1}) by the terms 
from the test set $\mathcal{TS}(\Set,\R)$, and obtain:
\begin{alignat}{2}
\sorted(\emptyset) =  \true \label{ex-i11}\\
\sorted(\ins(x_1,\emptyset)) =\true \label{ex-ii11}\\
\sorted(\ins(x_1,\ins(x_2, y))) =\true \label{ex-iii11}
\end{alignat}

Subgoals~(\ref{ex-i11}) and~(\ref{ex-ii11}) are simplified by $\RD$
(respectively with rules~(\ref{eq:sorted-empty}) 
and~(\ref{eq:sorted-single})) into $\true = \true$ which is a tautology.
Subgoal~(\ref{ex-iii11}) cannot be simplified by $\RD$,
because of the constraints in rewrite rules.
Subgoal~(\ref{ex-iii11}) does not contain any induction variable, 
and therefore, it cannot be further instantiated. 
So, the proof stops without a conclusion.
Hence, we fail to prove Conjecture~(\ref{ex-goal1}) with test set induction technique.

Concerning Conjecture~(\ref{ex-goal-occur}), the specification of the rules 
for $\occurbis$ contains membership constraints. 
This kind of specification is not supported by the current test-set induction procedures.

\subsection{Constrained Grammars based Induction} \label{sec:ex-grammars}
As discussed above, we need to add appropriate constraints while instantiating
the induction goals. This is precisely what constrained tree grammars do. 

Our procedure, presented in Section~\ref{se-inference},  roughly works as follows:
given a conjecture $C$
we try to apply the production rules of the normal form grammar to $C$ 
(instead of instantiating by terms of a test set)
as long as the depth of the clauses obtained is smaller or equal to the maximal
depth of a left-hand-side of $\RD$.
All clauses obtained must be reducible by $\R$,
or by induction hypotheses or either by others conjectures
not yet proved and smaller than $C$.
If this succeeds, the clauses obtained after simplification
are considered as new subgoals 
and for their proof we can use $C$ as an induction hypothesis.
Otherwise, the procedure fails and 
we have established a disproof under some assumptions on $\R$.

In order to prove Conjecture~(\ref{ex-goal1}),
we constraint the variable $y$ of this clause 
to belong to one of the languages defined
by non-terminals (of a compatible sort) of the normal form grammar $\G_\NF(\RC)$.
This is not restrictive since $\RC$ is terminating and $\R$ is sufficiently complete.
\begin{align}
 sorted(y) & = \true \cons[y \deriv \dxstate{x}{\Set}]          
 \tag{\ref{ex-goal1}.a}  \label{ex-decore1}\\
 sorted(y) & = \true  \cons[y \deriv \state{\ins(x_1, y_1)}]
 \tag{\ref{ex-goal1}.b} \label{ex-decore2}
\end{align}

Let us apply the above principle to the proof of Conjecture~(\ref{ex-goal1}).
The application of the production rules of the grammar 
to~(\ref{ex-decore1}) and~(\ref{ex-decore2}) returns:
\begin{align}
sorted(\emptyset) & =  \true
 \tag{\ref{ex-i11}'}\\
\sorted(\ins(x_1, \emptyset)) & = \true  \cons[x_1 \deriv \dxstate{x}{\Nat}]
  \tag{\ref{ex-ii11}'} \\
\sorted(\ins(x_1, \ins(x_2, \emptyset))) & = \true \cons[x_1, x_2 \deriv \dxstate{x}{\Nat}, x_1 \prec x_2] 
  \tag{\ref{ex-iii11}'} \\
\sorted(\ins(x_1, \ins(x_2, y_2))) & = \true \tag{\ref{ex-iii11}''}\\
 & \cons[x_1, x_2, x_3 \deriv \dxstate{x}{\Nat}, y_2 \deriv \state{\ins(x_3, y_3)}, 
         x_1 \prec x_2, x_2 \prec x_3] \notag
\end{align}
For obtaining~(\ref{ex-ii11}'), (\ref{ex-iii11}') and (\ref{ex-iii11}''),
several steps of application of the production rules
of the grammar are necessary.
Subgoals~(\ref{ex-i11}'),~(\ref{ex-ii11}')  are simplified by $\RD$ into a tautology, 
like in Section~\ref{sec:ex-test-set}.
Unlike Section~\ref{sec:ex-test-set}, Subgoal~(\ref{ex-iii11}') can now be simplified
using the rule~(\ref{eq:sorted-2})  of $\RD$, because of its constraint $x_1 \prec x_2$.
Moreover, Subgoal~(\ref{ex-iii11}'') can be reduced by 
the rule~(\ref{eq:sorted-2}) into:
\[ \sorted(\ins(x_2, y_2)) = \true 
   \cons[x_2, x_3 \deriv \dxstate{x}{\Nat}, y_2 \deriv \state{\ins(x_3, y_3)}, x_2 \prec x_3] \]
This latter subgoal can be itself simplified
into $\true = \true$ by~(\ref{ex-goal1}),
used here as an induction hypothesis.
This terminates the inductive proof of~(\ref{ex-goal1}).

\bigskip

For the proof of Conjecture~(\ref{ex-goal-occur}), the situation is more complicated.
The decoration of the variables of~(\ref{ex-goal-occur}) with non terminals
of the grammar $\G_\NF(\RC)$ returns:
\begin{align}
x \occurbis y & = x \occur y 
  \cons[x \deriv \dxstate{x}{\Nat}, y \deriv \dxstate{x}{\Set} ]
  \tag{\ref{ex-goal-occur}.a}  \label{ex2-deco1}\\
x \occurbis y & = x \occur y 
  \cons[x \deriv \dxstate{x}{\Nat}, y \deriv \state{\ins(x_1, y_1)} ]
  \tag{\ref{ex-goal-occur}.b} \label{ex2-deco2}
\end{align}
The application of the production rules of $\G_\NF(\RC)$ to these clauses gives:
\begin{align}
x \occurbis \emptyset & =  x \occur \emptyset 
  \label{ex2-i} \\ 
x \occurbis \ins(x_1, \emptyset) & =  x \occur \ins(x_1, \emptyset) 
  \cons[x, x_1 \deriv \dxstate{x}{\Nat}]
  \label{ex2-ii}\\ 
x \occurbis \ins(x_1,  \ins(x_2, \emptyset)) & =  x \occur \ins(x_1,  \ins(x_2, \emptyset))
 \cons[x, x_1, x_2 \deriv \dxstate{x}{\Nat}, x_1 \prec x_2]  
 \label{ex2-iii} \\ 
x \occurbis \ins(x_1, \ins(x_2, y_2)) & =  x \occur \ins(x_1,  \ins(x_2, y_2)) \notag \\
 & \cons[x, x_1, x_2 \deriv \dxstate{x}{\Nat}, y_2 \deriv \state{\ins(x_3, y_3)}, x_1 \prec x_2, x_2 \prec x_3]  
 \label{ex2-iv} 
\end{align}
The clause~(\ref{ex2-i})  is reduced, 
using~(\ref{eq:occur2-empty}) and~(\ref{eq:occur-empty}),
to the tautology $\false = \false$.

In order to simplify (\ref{ex2-ii}), 
we restrict to the cases corresponding to the constraints of
the rules~(\ref{eq:occur2-eq}),~(\ref{eq:occur2-leq}) and~(\ref{eq:occur2-geq}).
This technique, called \emph{Rewrite Splitting}, is defined formally in Section~\ref{se-inference}.
We obtain respectively:
\begin{align}
\true & =  x \occur \ins(x_1, \emptyset) \cons[x_1 \deriv \dxstate{x}{\Nat}, x \approx x_1]
  \tag{\ref{ex2-ii}.1} \label{ex2-ii.1}\\ 
\false & =  x \occur \ins(x_1, \emptyset)
  \cons[x_1 \deriv \dxstate{x}{\Nat}, \ins(x_1, \emptyset) \deriv \state{\ins(x_2, y_2)}, x \prec x_2]
  \tag{\ref{ex2-ii}.2} \label{ex2-ii.2}\\ 
x \occurbis \emptyset & =  x \occur \ins(x_1, \emptyset) \cons[x_1 \deriv \dxstate{x}{\Nat}, x_1 \prec x]
  \tag{\ref{ex2-ii}.3} \label{ex2-ii.3}
\end{align}
Note that the constraint in (\ref{ex2-ii.2}) implies that $x_1 = x_2$.
All these subgoal are reduced into tautologies $\true = \true$ or $\false = \false$
using respectively the following rules of $\RD$:
\begin{itemize}
\item (\ref{eq:occur-eq}) for (\ref{ex2-ii.1}), 
\item (\ref{eq:occur-neq}) and~(\ref{eq:occur-empty}) for (\ref{ex2-ii.2}) (with $x_1 = x_2$),  
\item (\ref{eq:occur2-empty}) for the left member of (\ref{ex2-ii.3}),
and  (\ref{eq:occur-neq}) then~(\ref{eq:occur-empty}) for its right member.
\end{itemize}

The subgoal (\ref{ex2-iii}) 
is also treated by \emph{Rewrite Splitting} with the
rules (\ref{eq:occur2-eq}), (\ref{eq:occur2-leq}), (\ref{eq:occur2-geq}) of $\RD$, 
similarly as above. 


Let us now finish the proof of Conjecture~(\ref{ex-goal-occur}), 
with the subgoal~(\ref{ex2-iv}).
By rewrite splitting with the
rules (\ref{eq:occur2-eq}), (\ref{eq:occur2-leq}), (\ref{eq:occur2-geq}), we obtain:
\begin{multline}
\true  =  x \occur \ins(x_1,  \ins(x_2, y_2)) \notag \\
  \cons[x, x_1, x_2, x_3 \deriv \dxstate{x}{\Nat}, y_2 \deriv \state{\ins(x_3, y_3)}, 
             x_1 \prec x_2, x_2 \prec x_3, x \approx x_1]  
 \tag{\ref{ex2-iv}.1} \label{ex2-iv.1} 
\end{multline}
\begin{multline}
\false  =  x \occur \ins(x_1,  \ins(x_2, y_2)) \notag \\
\left\lcons
\begin{array}{l} 
  x, x_1, x_2, x_3, x_4 \deriv \dxstate{x}{\Nat}, y_2 \deriv \state{\ins(x_3, y_3)}, x_1 \prec x_2, x_2 \prec x_3, \\
  \ins(x_1,  \ins(x_2, y_2)) \deriv \state{\ins(x_4, y_4)}, x \prec x_4
\end{array}%
\right\rcons                
 \tag{\ref{ex2-iv}.2} \label{ex2-iv.2} 
\end{multline}
\begin{multline}
x \occurbis \ins(x_2, y_2)  =  x \occur \ins(x_1,  \ins(x_2, y_2)) \notag \\
  \cons[x, x_1, x_2, x_3 \deriv \dxstate{x}{\Nat}, y_2 \deriv \state{\ins(x_3, y_3)}, x_1 \prec x_2, x_2 \prec x_3, x_1 \prec x]  
 \tag{\ref{ex2-iv}.3} \label{ex2-iv.3} 
\end{multline}

The subgoal~(\ref{ex2-iv.1}) is simplified by (\ref{eq:occur-eq}) into the tautology $\true = \true$.

The subgoal~(\ref{ex2-iv.3}) is simplified by (\ref{eq:occur-neq}) into:
\begin{equation}
 x \occurbis \ins(x_2, y_2)  =  x \occur \ins(x_2, y_2)
  \cons[x, x_2, x_3 \deriv \dxstate{x}{\Nat}, y_2 \deriv \state{\ins(x_3, y_3)}, 
             x_2 \prec x_3 ] 
\label{ex2-iv.3.1}
\end{equation}             
At this point, we are allowed to use the goal~(\ref{ex-goal-occur}) as an induction hypothesis 
since we have perform a reduction step on the subgoals.
A simplification of~(\ref{ex2-iv.3.1}) using~(\ref{ex-goal-occur}) gives the tautology:
\[
 x \occurbis \ins(x_2, y_2)  =  x \occurbis \ins(x_2, y_2)
  \cons[x, x_2, x_3 \deriv \dxstate{x}{\Nat}, y_2 \deriv \state{\ins(x_3, y_3)}, 
             x_2 \prec x_3 ] 
\]

For the subgoal~(\ref{ex2-iv.2}), note that in the constraints, $x \prec x_4$ implies $x \prec x_1$.
Hence~(\ref{ex2-iv.2}) can be simplified by~(\ref{eq:occur-neq}) into:
\( 
\false  =  x \occur \ins(x_2, y_2)   \cons[\ldots] 
\).
A simplification of the above subgoal using~(\ref{ex-goal-occur}) (as an induction hypothesis) 
gives:
\( \false  =  x \occurbis \ins(x_2, y_2) \cons[\ldots] \).
The above subgoal has the same constraints as~(\ref{ex2-iv.2}),
and it can be observed that this constraint implies $x \prec x_2$. 
Therefore, 
we can simplify this subgoal using~(\ref{eq:occur2-leq}) into the tautology $\false = \false$.

In conclusion, Conjecture~(\ref{ex-goal-occur})  can be proved with our approach based on constrained grammars without
the addition of any lemmas.

\subsection{Proof with ACL2}
\label{sec:acl2}
A proof of Conjecture~(\ref{ex-goal-occur}) was done by 
Jared Davis\footnote{Jared Davis, personal communication.}
with the \textsf{ACL2} theorem prover, using his library \textsf{osets} for finite set theory~\cite{Davis-osets}.
In this library, sets are implemented on fully ordered lists (wrt an ordering \verb!<<!).
The definition in \textsf{osets} of a function \verb!insert a X!,
for insertion of an element \verb!a! to a list \verb!X! is the same as the above axioms of $\RC$:
\begin{verbatim}
(defun insert (a X)
  (declare (xargs :guard (setp X)))
  (cond ((empty X) (list a))
        ((equal (head X) a) X)
        ((<< a (head X)) (cons a X))
        (t (cons (head X) (insert a (tail X))))))
\end{verbatim}
It refers to the functions \verb!head!  and \verb!tail!
which return respectively the first (smallest) element in list 
(the LISP \verb!car!) and the rest of a list (LISP \verb!cdr!).
The guard \verb!(setp X)! ensures that \verb!X! is 
a fully ordered list without duplication.

The library \textsf{osets} contains a definition of membership 
similar to the axioms of 
(\ref{eq:occur-empty}--\ref{eq:occur-neq}) of $\RD$ for the definition of $\in$:
\begin{verbatim}
(defun in (a X)
  (declare (xargs :guard (setp X)))
  (and (not (empty X))
       (or (equal a (head X))
	   (in a (tail X)))))
\end{verbatim}

Next, our defined function $\occurbis$ becomes the following \verb$inb$:
\begin{verbatim}
(defun inb (a X)
 (declare (xargs :guard (setp X)))
 (and (not (empty X))
      (not (and (setp X) (<< a (head X))))
      (or (equal a (head X))
          (inb a (tail X)))))
\end{verbatim}

\noindent The conjecture~(\ref{ex-goal-occur}) becomes:
\begin{verbatim}
(defthm in-is-inb
 (equal (in a X)
        (inb a X)))
\end{verbatim}

Using the \textsf{osets} library, the system proved everything except the following subgoal:
\begin{verbatim}
   (IMPLIES (AND (NOT (EMPTY X))
                 (SETP X)
                 (<< A (HEAD X)))
            (EQUAL (IN A X) (INB A X))).
\end{verbatim}

The following lemma permits to finish the proof:
\begin{verbatim}
  (defthm head-minimal
    (implies (<< a (head X))
	     (not (in a X)))
    :hints(("Goal" 
	    :in-theory (enable primitive-order-theory))))
\end{verbatim}

The lemma \verb!head-minimal!
was not available to users of the library \textsf{osets}.
It will be incorporated (together with the technical lemma for its proof) 
in the  appropriate file of the \textsf{osets} library.
%
\begin{verbatim}
(local (defthm lemma
	   (implies (and (not (empty X))
			 (not (equal a (head X)))
			 (not (<< a (head (tail X))))
			 (<< a (head X)))
		    (not (in a X)))
	   :hints(("Goal" 
		   :in-theory (enable primitive-order-theory)
		   :cases ((empty (tail X)))))))
\end{verbatim}

Note that this proof uses several theorems and hints
included in the \textsf{osets} library.
Without this library, the \textsf{ACL2} theorem prover
would need the addition of several key lemmas and hints.
For finding them, the user would be required both
experience and a good understanding of the problem and how to solve it.

\subsection{Assisted Proof with SPIKE}
\label{sec:spike}
Conjecture~(\ref{ex-goal-occur}) was proved 
with the last version of \textsf{SPIKE}
by Sorin Stratulat\footnote{Sorin Stratulat, personal communication.}

Since \textsf{SPIKE} does not support constrained axioms, 
constraints are expressed as conditions.
The specification of $\sorted$ becomes:
\begin{verbatim}
sorted(Nil) = true;
sorted(ins(x, Nil)) = true;
x1 <= x2 = true => sorted(ins(x1, ins(x2, y))) = sorted(ins(x2, y));
x1 <= x2 = false => sorted(ins(x1, ins(x2, y))) = false;
\end{verbatim}

\noindent The axioms for $\occur$ and $\occurbis$ are respectively:
\begin{verbatim}
in(x1, Nil) = false;
x1 = x2 => in(x1, ins(x2, y)) = true;
x1 <> x2 => in(x1, ins(x2, y)) = in(x1, y);
\end{verbatim}
\noindent and
\begin{verbatim}
in'(x1, Nil) = false;
x1 = x2 => in'(x1, ins(x2, y)) = true;
x2 < x1 = true => in'(x1, ins(x2, y)) = in'(x1, y);
x1 < x2 = true, osetp(ins(x2,y)) = true => in'(x1, ins(x2, y)) = false;
x1 < x2 = true, osetp(ins(x2,y)) = false => in'(x1, ins(x2, y)) = in'(x1, y);
\end{verbatim}

\noindent The unary predicate \verb|osetp| characterizes 
ordered lists. It is defined by the following axioms.
\begin{verbatim}
osetp(Nil) = true;
osetp(ins(x, Nil)) = true;
osetp(ins(x, ins(y, z))) = and(x < y, osetp(ins(y, z)));
\end{verbatim}

With this predicate, the conjecture is expressed as follows.
\begin{verbatim}
osetp(y) = true => in(x, y) = in'(x, y);
\end{verbatim}

A particular user specified strategy and
the following additional lemmas were necessary 
for the termination of the proof with \textsf{SPIKE}.
The three first lemma are natural, the last one is less intuitive.
\begin{verbatim}
osetp(y) = true => sorted(y) = true;
osetp(ins(u1, u2)) = true => osetp(u2) = true;
u1 < u2 = true, u2 < u3 = true => u1 < u3 = true;
osetp(ins(u4, u5)) = true , u2 < u4 = true  => in(u2, u5) = false;
\end{verbatim}

\subsection{Verification of Trace Properties} \label{sec:protocols}
We have seen in the previous sections how membership constraints
can be used in the axioms of $\R$ for the specification of
operations on complex data structures, and how our method can handle it.
Our procedure can also handle membership constraints in the conjecture.
This feature can be used for instance in order to  restrict some terms to a particular pattern.
It is very useful in the context of the verification of infinite systems,
in order to express that a trace of events belongs to a (regular) set of bad traces.

In~\cite{BouhoulaJacquemard07arspa} we follow this approach 
for the verification of security properties of cryptographic protocols,
using an adaptation of the procedure of this paper in order to deal with
specifications which are not necessarily confluent and sufficiently complete.
In this section we wont describe in full details the specification of~\cite{BouhoulaJacquemard07arspa} 
but we shall roughly describe the main lines of the approach.
Consider the following conjecture:
\begin{equation} \label{eq:trace}
\mathit{trace}(y) \neq \mathit{true} \cons[ y \deriv \dxstate{x}{\List}, y \deriv  \dxstate{x}{\bad}] 
\end{equation}
Here, the membership constraint $y \deriv \dxstate{x}{\List}$ restricts $y$ to be generated by the
non terminal $\dxstate{x}{\List}$ of the normal form constrained tree grammar.
It means that $y$ is a constructor term  in normal form (as in the above example of sorted lists)
representing a list of events of a system.
The second membership constraint $y \deriv  \dxstate{x}{\bad}$
further restricts $y$ to belong to a regular tree language representing faulty traces
(traces which lead to a state of the system corresponding to a failure, an attack for instance).
Finally the clause $\mathit{trace}(y) \neq \mathit{true}$ expresses that $y$ is not 
a trace of the system.
Hence the above conjecture~(\ref{eq:trace}) means that every bad trace is not reachable. 

The defined function $\mathit{trace}$ can be  specified using constrained conditional rewrite rules.
For instance, in~\cite{BouhoulaJacquemard07arspa}, 
we follow the approach of Paulson~\cite{Paulson98} for the inductive 
specification of the messages exchanges of the protocol,
and of the actions of the insecure communication environment.
Note also that we extends this model with equations specifying
the cryptographic operations, like the following non-left-linear equation for the decryption operator $\sdec$ in a symmetric
cryptosystem:
\(
\sdec\bigl(\senc(x, y), y\bigr ) \to x 
\).
These axioms, sometimes referred as  \emph{explicit destructors} equations, 
permit a strict extension of the verification model 
(they allow strictly more attacks on protocols)
and they are specified as constructor equations of $\RC$ in our model.

\end{RR}


\section{Constrained Tree Grammars} \label{se-grammar}

Constrained tree grammars have been introduced in~\cite{ComonThese88}, 
in the context of automated induction.
The idea of using such formalism for induction theorem proving 
is also in \textit{e.g.}~\cite{BouhoulaJouannaud01,Comon01handbook},
because it is known that they can generate the languages of normal-forms
for arbitrary term rewriting systems.

\begin{RR}
In this paper, we push the idea one step beyond with 
a full integration of tree grammars with constraints in our 
induction procedure.
Indeed, constrained tree grammars are used here:
\begin{itemize}
\item[i.] as an induction scheme (instead of test-sets), 
for triggering induction steps by instantiation of subgoals using production rules,
\item[ii.] as a decision procedure for checking deletion criteria,
including tests like ground irreducibility or validity in restricted cases, 
as long as emptiness is decidable.
\item[iii.] for  the definition and treatment of constraints of membership in fixed tree languages,
in particular languages of normal forms.
\end{itemize}
\end{RR}

\noindent We present in this section the definitions and results suited to our purpose.
\begin{definition} \label{def:grammar}
A \emph{constrained grammar} $\G = (Q, \Delta)$ is given by:
1. a finite set $Q$ of non-terminals of the form $\state{u}$, where $u$ is a linear term of $\T(\F,\X)$,
2. a finite set $\Delta$ of production rules of the form 
\(\state{v} := f(\state{u_1},\ldots,\state{u_n}) \cons[c]\) 
where  $f \in \F$, $\state{v}$, $\state{u_1}$,\ldots, $\state{u_n} \in Q$ 
(modulo variable renaming) and $c$ is a constraint.
\end{definition}
The non-terminals are always considered modulo variable renaming.
In particular, we assume \emph{wlog} (for technical convenience)
that the above term $f(u_1,\ldots,u_n)$ is linear
and that  $\var(v) \cap \var(f(u_1,\ldots,u_n)) = \emptyset$.

\subsection{Languages of Terms} \label{sb-production}
We associate to a given constrained grammar $\G = (Q, \Delta)$ a finite set
of new unary predicates of constraint of the form $.\, \deriv \state{u}$, 
where $\state{u} \in Q$ (modulo variable renaming). 
Constraints of the form $t \deriv \state{u}$  called
\emph{membership constraints} and their interpretation is given below.
The production relation between constrained terms $\vdash^y_{\G}$ 
is defined by:
\[ 
t[y] \cons[y \deriv \state{v} \land d] 
\vdash^y_{\G}
t[f(y_1,\ldots,y_n)] 
\cons[ y_1 \deriv \state{u_1}  \land \ldots \land y_n \deriv \state{u_n} \land c \land d\tau ] 
\] 
if there exists $\state{v} := f(\state{u_1},\ldots,\state{u_n}) \cons[c] \in \Delta$  
such that $f(u_1,\ldots,u_n) = v\tau$, 
and $y_1$,\ldots,$y_n$ are fresh variables.
The variable $y$, constrained to be in the language defined by the non-terminal
$ \state{v}$ is replaced by $f(y_1,\ldots, y_n)$ where the variables $y_1,\ldots, y_n$
are constrained to the respective languages of non-terminals  $\state{u_1},\ldots,\state{u_n}$.
The union of the relations $\vdash^y_\G$ for all $y$ is 
denoted $\vdash_\G$ and
the reflexive transitive and transitive closures of the relation
$\vdash_\G$ are respectively denoted by $\vdash_\G^*$ and~$\vdash_\G^+$
($\G$ may be omitted).
\begin{definition} \label{def-generated} \label{def:language}
The language $L(\G,\state{u})$ is the set of ground terms $t$
\emph{generated} by 
\begin{RR}
a constrained grammar 
\end{RR}
$\G$ from a non-terminal $\state{u}$, 
\textit{i.e.} 
such that $y\cons[y \deriv \state{u}] \vdash^* t\cons[c]$ where $c$ is satisfiable. 
\end{definition}
\noindent Given $Q' \subseteq Q$, 
we write $L(\G, Q') = \bigcup_{\state{u}\in Q'}L(\G,\state{u})$ 
and $L(\G) = L(\G, Q)$.
%
%
\noindent Given a constrained grammar $\G = (Q, \Delta)$, 
we can now define $\sol(t \deriv \state{u})$, 
where $\state{u} \in Q$,
as $\{ \sigma \mid t\sigma \in L(\G,\state{u}) \}$. 

\begin{ABS}
\begin{example} \label{ex:grammar-nat}
Let us consider the sort $\Nat$ of natural integers built with the
constructor symbols $0$ and $s$.
These terms are generated by the grammar with production rules
$\xstate{x}{\Nat} := 0$ and
$\xstate{x}{\Nat} := s(\xstate{x_2}{\Nat})$.\finex
\end{example}
\end{ABS}

\begin{RR}
\begin{example}
With the normal grammar of 
Section~\ref{sec:ex-grammars},
denoted $\G$ in this example, we have:
$L(\G, \xstate{x}{\Bool}) = \{ \true, \false \}$, 
$L(\G, \xstate{x}{\Nat}) = \{ 0, s^n(0) \mid n > 0 \}$, 
$L(\G, \xstate{x}{\Set}) = \{ \emptyset \}$, 
$L(\G, \state{\ins(x_1, x_2)}) = \{ \ins(s^{n_1}(0), \ins(\ldots, \ins(s^{n_k}(0)))) \mid
  k \geq 1, n_1 < \ldots < n_k \}$, 
\finex
\end{example}
\end{RR}

\begin{RR}
Note that every regular tree language $L$ can be generated by a constrained tree grammar 
following Definitions~\ref{def:grammar} and~\ref{def:language},
with production rules of the form:
$\dxstate{x}{S} := f(\dxstate{x}{S_1}, \ldots, \dxstate{x}{S_1})$ 
where $S_1, \ldots, S_n, S$ are new sorts representing the non terminals 
of a regular tree grammar generating $L$.
\end{RR}

\begin{RR}
The intersection between the language generated by a constrained tree grammar 
(in some non-terminal) and a regular tree language
is generated by a constrained tree grammar.
The constrained grammar for the intersection is built with a product construction.
\end{RR}

\subsection{Languages of Normal Forms} \label{sb-nf}
\begin{ABS}
In~\cite{BouhoulaJacquemard08rrlsv07},
we present the automatic construction of a constrained grammar
$\G_\NF(\RC) = (Q_\NF(\RC),\Delta_\NF(\RC))$
which
\end{ABS}
\begin{RR}
The constrained grammar $\G_\NF(\RC) = (Q_\NF(\RC),\Delta_\NF(\RC))$ 
defined in Figure~\ref{nfgf}
\end{RR}
generates the language of ground $\RC$-normal forms.
Its construction is a generalization of the one 
of~\cite{ComonJacquemard03}.
Intuitively, it corresponds to the complementation and completion 
of a grammar for $\RC$-reducible terms
(such a grammar does mainly pattern matching of left members of rewrite rules),
where every subset of states (for the complementation)
is represented by the most general common instance of its elements (if they are unifiable).
\begin{RR}
For purpose of the the construction of $\G_\NF(\RC)$, a new sort $\re$ is added 
to $\mathcal{S}$,  
(the sort of reducible terms), 
and hence also a new variable $x^\re$.
\end{RR}
\begin{RR}
\begin{mafigure}[ht]
\begin{tabular}{l} 
\quad 
\(
\mathcal{L}(\RC)  =  
\left\{
\begin{array}{rcl}
 {u} & | &  u \mbox{~is~a~strict~subterm~of~} l \mbox{~for  some~} l \to r \cons[c] \in \RC \\
 & & \mbox{or~} u \mbox{~is~a~subterm~of~} l \mbox{~if~} c \mbox{~is empty} 
\end{array}
\right\} 
\) \\      
\quad 
\(
Q_\NF(\RC) = \left\{     
\begin{array}{l}     
\state{\mgi(t_1,\ldots,t_n)} \mid \{ {t_1},\ldots,{t_n} \} \mbox{~is a maximal}\\       
\multicolumn{1}{r}{\mbox{subset of~} \mathcal{L}(\RC) \mbox{~s.t.~} t_1,\ldots,t_n \mbox{~are unifiable}}       
\end{array}  
\right\} 
\uplus \bigl\{ \xstate{x}{{S}} \bigm| S \in \mathcal{S} \bigr\}
\)\\  
~\\  
\quad $\Delta_\NF(\RC)$ contains:\\   
\quad every  $\xstate{x}{\re} := f(\state{u_1},\ldots,\state{u_n}) \cons$ such that one
 of the $u_i$ at least is $x^\re$,\\     
\quad every $\xstate{x}{\re} := f(\state{u_1},\ldots,\state{u_n}) \cons[c]$ and
 every $\state{t} := f(\state{u_1},\ldots,\state{u_n}) \cons[\neg c]$\\     
\quad\quad  such that  $f \in \F$ with profile $S_1,\ldots,S_n \to S$\\        
\quad\quad and $\state{u_1} ,\ldots, \state{u_n} \in Q_\NF(\RC)$, $u_1,\ldots,u_n$
 have respective sorts $S_1,\ldots,S_n$\\      
\quad\quad \( t = \mgi\bigl\{ u \bigm| \state{u} \in Q_\NF(\RC)\mbox{~and~}
                                              u \mbox{~matches~} f(u_1,\ldots,u_n) \bigr\}\)\\ 
\quad\quad \( \displaystyle c \equiv \bigvee_{\scriptstyle l \to r  \cons[e] \in \RC,\  \scriptstyle f(u_1,\ldots,u_n) = l\theta } e\theta \)
\end{tabular} 
\caption{Constrained grammar $\G_\NF(\RC)$  for $\RC$-normal forms  \label{nfgf} \label{fi-nf}} 
\end{mafigure}
\end{RR}
\begin{ABS}
Due to space limitations, we cannot describe the general construction of $\G_\NF(\RC)$
here but we rather present the case where $\RC$ contains the constructor axioms 
given in introduction.
\begin{example} \label{ex:grammar-nf}
Let $\RC$ contain the axioms (\ref{eq:ins-double}) and (\ref{eq:ins-ord}) given in introduction.
Let $\G_\NF(\RC)$ contain the two productions rules given Example~\ref{ex:grammar-nat}
and
$\xstate{x}{\Set}  := \emptyset$,
$\state{\ins(x, y)} := \ins(\xstate{x}{\Nat},\xstate{x}{\Set})$
(for singleton lists)
and 
$\state{\ins(x, y)} \quad := \quad 
\ins(\dxstate{x}{\Nat}, \state{\ins(x_2, y_2)})\cons[x^\Nat \prec x_2] $.
Note that the variables in the non terminal $ \state{\ins(x_2, y_2)}$ in the right member
of the latter production rule have been 
renamed in order to be distinguished from the variables in the non terminal in the left member.
This grammar $\G_\NF(\RC)$ generates the set of ground constructor terms in normal-form for $\RC$.
They represent the ordered lists of natural numbers (of sort $\List$).  
\finex\end{example}
\end{ABS}
\begin{RR}
An example of a constrained grammar for $\RC$-normal forms constructed this way
was given in Section~\ref{sec:sorted-lists-grammar}.
\end{RR}

\begin{RR}
\begin{lemma} \label{lem:GNFRC}
For every term $t \in \T(\C)$, 
$t \in L(\G_\NF(\RC), \state{u})$ for some $\state{u} \in  Q_\NF(\RC) \setminus \{ \xstate{x}{\re} \}$ 
iff $t$ is an  $\RC$-normal form.
\end{lemma}
\begin{proof}
We shall use the following Fact, which can
be proved by a straightforward induction on the length
of the derivation $y\cons[y \deriv \state{u}] \vdash^* t\cons[c]$.
\begin{fact} \label{fac:GNFRC-nt}
For each $\state{u} \in  Q_\NF(\RC) \setminus \{ \xstate{x}{\re} \}$,
and each $t \in L(\G_\NF(\RC), \state{u})$, 
$t$ is an instance of $u$ and 
$u = \mgi\bigl\{ v \bigm| \state{u} \in  Q_\NF(\RC) \setminus \{ \xstate{x}{\re} \}
\mbox{~and~} $t$ \mbox{~is~an~instance~of~} v \bigr\}$.
\end{fact}

Let us now show the  'only if' direction
by induction on the length
of the derivation $y\cons[y \deriv \state{u}] \vdash^* t \cons[c']$
(where $c'$ is satisfiable).

If the length is $1$, then $t$ is a nullary symbol of $\C$,
and by construction $t$ is $\RC$-irreducible.

If \( y\cons[y \deriv \state{u}] \vdash 
     f(y_1,\ldots,y_n) \cons[y_1 \deriv \state{u_1\theta} \land
     \ldots \land y_n \deriv \state{u_n\theta} \land c\theta ]
     \vdash^* t \cons[c'] = f(t_1, \ldots,t_n)\cons[c'] \)
for some production rule 
$\state{u} := f(\state{u_1},\ldots,\state{u_n}) \cons[c] \in \Delta_\NF(\RC)$  
($\theta$ is a variable renaming by fresh variables),
then for every $i \in [1..n]$, $t_i \in L(\G_\NF(\RC), \state{u_i})$,
and $\state{u_i} \neq \xstate{x}{\re}$
(otherwise we would have $\state{u} = \xstate{x}{\re}$).
Hence, by induction hypothesis, every $t_i$ is a $\RC$-normal form.
Assume that $t$ is $\RC$-reducible (it must then be reducible at root
position),
and let $l \to r \cons[d] \in \RC$ be such that 
$t = l\tau$, $\tau \in \sol(d)$ and $l$ is maximum wrt subsumption
among the rules of $\RC$ satisfying these conditions.
By construction, $u = l$ and $c = \neg d\sigma \wedge c'$.
It follows from the satisfiability of $c'$ that $\tau \in \sol(c)$
(the variables of $c$ are instantiated by ground terms 
in the above grammar derivation).
This is in contradiction 
with $c = \neg d\sigma \wedge c'$ and $\tau \in \sol(d)$.

We show now the 'if' direction by induction on $t$.

If $t$ is a nullary function symbol of sort $S$ and is $\RC$-irreducible,
then $t$ is not the left-hand side of a rule of $\RC$, and
\( y\cons[y \deriv \xstate{x}{S}] \vdash t \).

If $t = f(t_1,\ldots,t_n)$ and is $\RC$-irreducible, then
every $t_i$ is $\RC$-irreducible for $i \in [1..n]$,
hence by induction hypothesis, 
$t_i \in L(\G_\NF(\RC), \state{u_i})$
for some $\state{u_i} \in  Q_\NF(\RC) \setminus \{ \xstate{x}{\re} \}$.
It means that for all $i \in [1..n]$,
there is a derivation of $\G_\NF(\RC)$ of the form
$y\cons[y \deriv \state{u_i}] \vdash^* t_i \cons[c_i]$.
By Fact~\ref{fac:GNFRC-nt}, every $t_i$ is an instance of $u_i$,
hence $t = f(u_1,\ldots,u_n) \tau$ for some ground substitution $\tau$.
If there is a production rule
$\state{u} := f(\state{u_1},\ldots,\state{u_n}) \cons[c] \in \Delta_\NF(\RC)$,
with $\state{u} \in  Q_\NF(\RC) \setminus \{ \xstate{x}{\re} \}$
and $\tau \in \sol(c)$,
then the following derivation is possible:
 \( 
 y\cons[y \deriv \state{u}] \vdash 
 f(y_1,\ldots,y_n) \cons[y_1 \deriv \state{u_1\theta} \land
 \ldots \land y_n \deriv \state{u_n\theta} \land c\theta ]
 \vdash^* t \cons[c'] 
\)
where $c'$ is satisfiable,
and $t \in L(\G_\NF(\RC), \state{u})$.
Assume that for every such production rule,
we have $\tau \not\in \sol(c)$.
It means by construction that
there is a rule $u \to r \cons[d] \in \RC$
such that $\tau \in \sol(d)$, hence 
that $t$ is $\RC$-reducible, a contradiction.
\hfill~$\Box$
\end{proof}

Using the observation that every ground constructor term is generated by 
$\G_\NF(\RC)$, we obtain as a corollary that  
$t \in L(\G_\NF(\RC), \xstate{x}{\re})$ iff $t$ is $\RC$-reducible.
\end{RR}


\section{Inference System}  \label{sec-inference} \label{se-inference}
In this section, we present an inference system 
for our inductive theorem proving procedure. 
\begin{ABS}
The principle is the following:
given a goal (conjecture) $C$, 
we use the grammar $\G_\NF(\RC)$ of Section~\ref{sb-nf}
in order to expand $C$ into some subgoals.
%
All the generated subgoals must then either be deleted, following some criteria,
or be reduced, using axioms or induction hypotheses,
or conjectures not yet proved, providing that they are smaller
than the goal to be proved. Reduced subgoals become then new conjectures
and $C$ becomes an induction hypothesis.

The deletion criteria 
include tautologies, forward subsumption, 
clauses with an unsatisfiable constraint, 
and constructor clause that can be detected as inductively valid,
under some conditions defined precisely below.
The decision of these criteria, using $\G_\NF(\RC)$, 
is discussed in Section~\ref{sb-scheme}.

The reduction of subgoals is performed 
with the rules defined in Sections~\ref{sb-srf} and~\ref{sb-src}.
If a subgoal generated cannot be deleted or reduced, then
the procedure stop with a refutation
(the initial goal is not an inductive theorem of $\R$).
If every subgoal is deleted, then the initial goal is an inductive theorem of $\R$.

The procedure may not terminate
(the conditions in inference rules other than the deletion criteria
are recursive calls of the procedure of the form $\R \modelsind \mathit{subgoal}$).
In this case appropriate lemmas should be added by the user in
order to achieve termination.

\end{ABS}
\begin{RR}
Let us first summarize the key steps of our procedure 
with the following pseudo-algorithm%
\footnote{Note that it is only a simplified version of the procedure, 
for presentation purpose, in order to give an intuition
of how the procedure operates.}.
The complete inference system,
introduced by the examples of Section~\ref{se-example},
is presented in details in Subsections~\ref{sb-srf}, \ref{sb-src} and~\ref{sb-inferences}.

\medskip
We start with a conjecture (goal) $G$ (a constrained clause) and a rewrite 
system (with conditions and constraints) $\R$, 
with a subset $\RC$ of constructor constrained (unconditional) rewrite rules.
\begin{enumerate}
\item \label{it-method-grammar} compute the constrained tree grammar
 $\G_\NF(\mathcal{R_C})$
\item \label{it-method-instanciate}
given a goal (or subgoal) $C$, 
  generate instances of $C$ by using the production rules of $\G_\NF(\mathcal{R_C})$.
  We obtain $C_1,\ldots, C_n$.
\item \label{it-method-subgoal} \label{it-method-loop}
\textbf{for each} $C_i$, do:
 \begin{enumerate}
 \item \label{it-method-validity} \label{it-method-deletion}
  \textbf{if} $C_i$ is a tautology 
  or $C_i$ is a constructor clause and can be detected as inductively valid 
  \textbf{then} delete it 
  \item \label{it-method-reduction}
  \textbf{else if} we are in  one of the two following cases:
  \begin{enumerate} 
  \item \label{it-method-reduction-constructor}
  	$C_i$ is a constructor clause and is reducible using $\RC$, \textbf{or} 
  \item \label{it-method-reduction-defined}
  	$C_i$ contains a non-constructor symbol and is reducible using $\R$ and induction hypotheses
  \end{enumerate}
  \textbf{then} reduce $C_i$ into $C'_i$
  \item \label{it-method-disproof}
  \textbf{else} \textbf{disproof} (the initial conjecture is not an inductive theorem)
\end{enumerate}
\item 
\textbf{if} \ref{it-method-subgoal} did not fail
\textbf{then} $C$ becomes an \emph{induction hypothesis}
\item \textbf{for each} $C'_i$, do:
\begin{enumerate}
 \item \label{it-method-redundancy}
 \textbf{if} $C'_i$ is a tautology or
  it is a constructor clause and  can be detected as inductively valid 
  or it is subsumed by an axiom or induction hypothesis
   \textbf{then} delete it 
 \item \textbf{otherwise} $C'_i$ becomes a new subgoal,
 \textbf{go to}~\ref{it-method-instanciate}.
\end{enumerate}
\end{enumerate}
If every subgoal is deleted, then $G$ is an inductive theorem of $\R$.
The procedure may not terminate, 
and in this case appropriate lemmas should be added by the user in
order to achieve termination.


The deletion criteria (steps~\ref{it-method-deletion} and~\ref{it-method-redundancy})
include tautologies, forward subsumption, 
clauses with an unsatisfiable constraint, 
and constructor clause and can be detected as inductively valid,
under some conditions defined precisely below.
The procedure for testing these criteria
is based on a reduction to a tree grammar non-emptiness problem
(does there exist at least one term generated by a given grammar),
using $\G_\NF(\mathcal{R_C})$.
In particular, it should be noted that we can decide validity this way 
for clauses $C_i$ which are ground irreducible~\cite{JouannaudKounalisLICS86,Kapur91acta}
(a notion central in inductive theorem proving / proof by consistency).
It is possible to decide ground irreducibility 
also by mean of reduction to non-emptiness, following the lines of~\cite{ComonJacquemard03}.
In Section~\ref{sb-scheme}, we show how such tests can be achieved effectively,
providing that $\R$ is ground confluent, 
for some classes of tree grammar with equality and disequality constraints
studied in former works~\cite{bogaert92stacs,CaronComonCoquideDauchetJacquemard94,dauchet95jsc,ComonJacquemard03}.
The extension to other kind of constraints
(like \textit{e.g.} ordering constraints) requires algorithms for 
corresponding classes of tree grammars 
(see discussions in Sections~\ref{se-powerlist} and~\ref{se-extension}).

The reductions at step~\ref{it-method-reduction}
are performed either with standard rewriting
or with \emph{ind. contextual rewriting} 
(case~\ref{it-method-reduction-defined})
or by case analysis,
(\emph{partial splitting} in case~\ref{it-method-reduction-constructor}
and \emph{rewrite splitting} in case~\ref{it-method-reduction-defined}).
These rules are defined formally in Sections~\ref{sb-srf} and~\ref{sb-src}.
\end{RR}

\begin{RR}
\subsection{Induction Ordering}
The inference and simplification rules below rely
on an ordering defined on the top of the following
complexity measure on clauses.
\begin{definition}
  The complexity of a constrained clause $C \cons[c]$ is the pair
  made of the two following components:
  $C$, ordered by the multiset extension of the ordering $>_e$ on literals,
and the number of constraints $d\sigma$ not occurring in $c$, such that
    there exists $l \to r \cons[d] \in \RC$ and $l\sigma$ is a subterm of $C$.
\end{definition}
We denote $\gg$ the ordering on constrained clauses defined as
the lexicographic composition of 
the orderings on the two components on the complexities.
\end{RR}

\subsection{Simplification Rules for Defined Functions} 
\label{sb-srf}
Our procedure uses the simplification rules for defined symbols 
presented in Figure~\ref{srdf}.
The rules in this figure define the relation $\lrstep{}{\Hyp}_\D$
for simplifying constrained clauses using $\RD$, $\R$ and a
given set $\Hyp$ of constrained clauses considered as induction hypotheses.

\textsf{Inductive~Rewriting} simplifies goals using the axioms of $\RD$
as well as instances of the induction hypotheses of $\Hyp$, 
provided that they are smaller than the goal. 
The underlying induction principle is based on 
\begin{ABS}
a 
\end{ABS}
\begin{RR}
the
\end{RR}
well-founded ordering $\gg$ on constrained clauses
\begin{ABS}
(see~\cite{BouhoulaJacquemard08rrlsv07}).
\end{ABS}
This approach is more general than structural induction
which is more restrictive concerning simplification with induction hypotheses
(see e.g.~\cite{BouhoulaRusinowitch95jar}).
\textsf{Inductive~Contextual~Rewriting} 
can be viewed as a generalization
of a rule in~\cite{Zhang-CTRS92} to handle constraints
by recursively discharging them as inductive conjectures.
\textsf{Rewrite~Splitting} simplifies a 
clause which contains a subterm matching some left member of rule of $\RD$.  
This inference checks moreover that all cases are covered 
for the application of 
$\RD$, 
\emph{i.e.} that for each ground substitution $\tau$, 
the conditions and the constraints of at least one rule is true wrt $\tau$.
Note that this condition is always true when $\R$ is sufficiently complete, 
and hence that this check is superfluous in this case.
\textsf{Inductive~Deletion} deletes tautologies and clauses with unsatisfiable
constraints.


\begin{mafigure}
\begin{tabular}{l}
\textsf{Inductive~Rewriting}:
 $\bigl\{C \cons[c] \bigr\} 
 \lrstep{}{\Hyp}_{\D} \bigl\{C' \cons[c] \bigr\}$
\\[1mm]
\multicolumn{1}{l}{%
\begin{tabular}{ll}
  if & $C \cons[c] \lrstep{}{\rho, \sigma} C' \cons[c]$,
  $l\sigma > r\sigma$ and $l\sigma > \Gamma\sigma$\\
  where & $\rho = \Gamma \Rightarrow l \to r \cons[c] \in \RD \cup 
  \{\psi\ \mid \psi \in \Hyp \mbox{~and~} C\cons[c] \gg \psi \}$\\
\end{tabular}
} 
\\
\\
\textsf{Inductive~Contextual~Rewriting}:
$\bigl\{ \Upsilon \Rightarrow C[l\sigma] \cons[c] \bigr\} 
 \lrstep{}{\Hyp}_{\D} \bigl\{ \Upsilon \Rightarrow C[r\sigma] \cons[c] \bigr\}$
\\[1mm]
\multicolumn{1}{l}{%
  if $\R \modelsind \Upsilon \Rightarrow \Gamma\sigma \cons[c \wedge c'\sigma]$,
  $l\sigma > r\sigma$   and $\{ l\sigma \} >^\mul \Gamma\sigma$,
  where $\Gamma \Rightarrow l \to r \cons[c'] \in \RD$ 
} 
\\
\\
\textsf{Rewrite~Splitting}:
$\bigl\{C[t]_p \cons[c] \bigr\}
 \lrstep{}{\Hyp}_{\D} 
\bigl\{ \Gamma_i\sigma_i \Rightarrow C[r_i\sigma_i]_p \cons[c \land c_i\sigma_i] \bigr\}_{i \in [1..n]}$
\\[1mm]
\begin{ABS}
  if $\R \modelsind \Gamma_1\sigma_1 \cons[c_1\sigma_1] \lor\ldots
  \lor \Gamma_n\sigma_n \cons[c_n\sigma_n]$,
  $t > r_i \sigma_i$ and $\{ t \} >^\mul \Gamma_i\sigma_i$,
  where the\\
  $\Gamma_i\sigma_i \Rightarrow l_i\sigma_i \to r_i\sigma_i \cons[c_i\sigma_i]$,
  $i \leq n$,
  are all the instances of rules in $\RD$
  such that $l_i\sigma_i=t$
\end{ABS}
\begin{RR}
\begin{tabular}{l}
  if $\R \modelsind \Gamma_1\sigma_1 \cons[c_1\sigma_1] \lor\ldots
  \lor \Gamma_n\sigma_n \cons[c_n\sigma_n]$,
  $t > r_i \sigma_i$ and $\{ t \} >^\mul \Gamma_i\sigma_i$\\
  where the 
  $\Gamma_i\sigma_i \Rightarrow l_i\sigma_i \to r_i\sigma_i \cons[c_i\sigma_i]$,
  $i \in [1..n]$\\
  are all the instances of  rules 
  $\Gamma_i \Rightarrow l_i \to r_i \cons[c_i] \in \RD$
  such that $l_i\sigma_i=t$
\end{tabular}
\end{RR}
\\
\\
\textsf{Inductive~Deletion}:
$\bigl\{C \cons[c] \bigr\}  \lrstep{}{\Hyp}_{\D} \emptyset$
if $C \cons[c]$ is a tautology or $c$ is unsatisfiable
\end{tabular} 
\caption{Simplification Rules for Defined Functions\label{srdf} \label{fi-defined}}
\end{mafigure}

\subsection{Simplification Rules for Constructors} 
\label{sb-src}

The simplification rules for constructors are presented in Figure~\ref{srcf},
they define the relation $\to_\C$ for simplifying constrained clauses using $\RC$ and $\R$.

\textsf{Rewriting} simplifies goals with axioms from $\RC$.
\textsf{Partial~Splitting} eliminates ground reducible terms in a constrained 
clause $C\cons[c]$ by adding to $C\cons[c]$ the negation of constraint of some rules
of $\RC$.  Therefore, the saturated application of \textsf{Partial~splitting} 
and \textsf{Rewriting} will always lead to
\textsf{Deletion} or to ground irreducible constructor clauses.
Finally, \textsf{Deletion} and \textsf{Validity} remove
respectively tautologies and clauses with unsatisfiable constraints,
and ground irreducible constructor theorems of $\R$.


\begin{mafigure}
\begin{tabular}{l}
\textsf{Rewriting}:
$\bigl\{C \cons[c] \bigr\} \to_{\C} \bigl\{ C' \cons[c] \bigr\}$
if $C \cons[c] \lrstep{+}{\RC} C' \cons[c]$ and $C\cons[c] \gg C'\cons[c]$
\\
\\
\textsf{Partial~Splitting}:
$ \bigl\{ C[l\sigma]_p \cons[c] \bigr\} \to_{\C} 
\bigl\{ C[r\sigma]_p \cons[c \land c'\sigma], C[l\sigma]_p \cons[c \land \lnot c'\sigma] \bigr\}$\\[1mm]
\multicolumn{1}{l}{%
if $l \to r \cons[c'] \in \RC$, $l\sigma > r\sigma$,
and neither $c'\sigma$ nor $\lnot c' \sigma$ is a subformula of $c$
} 
\\
\\
\textsf{Deletion}:
$\bigl\{ C \cons[c] \bigr\} \to_{\C} \emptyset$
if $C \cons[c]$ is a tautology or $c$ is unsatisfiable
\\
\\
\textsf{Validity}: 
$\bigl\{ C \cons[c] \bigr\} \to_{\C} \emptyset$\\[1mm]
 if $C\cons[c]$ is a ground irreducible constructor clause
 and $\R \modelsind C \cons[c]$ 
\end{tabular} 
\caption{Simplification Rules for Constructors\label{srcf} \label{fi-constructors}\label{fi-constructor}}
\end{mafigure}

%
%


\begin{mafigure}[t]
\begin{tabular}{l}
\textsf{Simplification}: 
\(
\Farc{ \bigl(\E \cup \bigl\{ C \cons[c] \bigr\},\Hyp \bigr)}%
     {\bigl(\E \cup \E',\Hyp \bigr)}
\)
if $\bigl\{ C \cons[c] \bigr\} \to_{\C} \E'$
\\
\\
\textsf{Inductive~Simplification}:
\(
\Farc{ \bigl(\E \cup \bigl\{ C \cons[c] \bigr\},\Hyp \bigr)}%
     {\bigl(\E \cup \E',\Hyp \bigr)}
\)
if 
$\{C \cons[c]\} \lrstep{}{\E\cup \Hyp}_{\D} \E'$
\\
\\
\textsf{Narrowing}:
\(
\Farc{\bigl(\E \cup \bigl\{ C \cons[c] \bigr\},\Hyp\bigr)}%
     {\bigl(\E \cup \E_1 \cup \ldots \cup \E_n, \Hyp \cup \{C\cons[c]\}\bigr)}
\)
\\[1mm]
\multicolumn{1}{l}{%
\begin{tabular}{l}
  if $\bigl\{ C_i \cons[c_i]\bigr\}  \to_{\C} \E_i$,
 where $\{ C_1\cons[c_1],\ldots,C_n\cons[c_n] \}$
  is the set of all clauses such that\\
  $C\cons[c] \vdash^{\ast} C_i \cons[c_i]$
  and
  $\depth(C_i) - \depth(C) \leq \depth(\R)-1$
\end{tabular}
}
\\
\\
\textsf{Inductive~Narrowing}:
\(
\Farc{\bigl(\E \cup \bigl\{ C \cons[c] \bigr\}, \Hyp \bigr)}
     {\bigl( \E \cup \E_1 \cup \ldots \cup \E_n, \Hyp \cup \{C\cons[c]\} \bigr)}
\)
\\[1mm]
\multicolumn{1}{l}{%
\begin{tabular}{l}
  if 
  $\{ C_i \cons[c_i] \}
    \lrstep{}{\E\cup \Hyp \cup  \{C\cons[c]\}}_{\D} \E_i$,
  where $\{ C_1\cons[c_1],\ldots,C_n\cons[c_n] \}$ is the set\\[1mm]
  of all clauses
  such that $C\cons[c] \vdash^+ C_i \cons[c_i]$
  and
  $\depth(C_i) - \depth(C) \leq \depth(\R)-1$ 
\end{tabular}
} 
\\
\\
\textsf{Subsumption}:
\(
\Farc{ \bigl(\E \cup \bigl\{ C \cons[c] \bigr\},\Hyp \bigr)}
     {(\E,\Hyp)}
\)
 if $C \cons[c]$  is subsumed by another clause of $\R \cup \E \cup \Hyp$
\\
\\
\textsf{Disproof}:
\(
\Farc{\bigl( \E \cup \bigl\{ C \cons[c] \bigr\}, \Hyp\bigr)}
     {(\bot,\Hyp)}
\)
if no other rule applies to the clause $C \cons[c]$
\end{tabular}
\caption{Induction Inference Rules\label{iirf} \label{fi-inference}}
\end{mafigure}


\subsection{Induction Inference Rules}  \label{sb-inferences}
The main inference system is displayed in Figure~\ref{iirf}.
Its rules apply to pairs $(\E,\Hyp)$ 
whose components are respectively the sets of current conjectures 
and of inductive hypotheses.
Two inference rules below,
\textsf{Narrowing} and \textsf{Inductive Narrowing}, use 
the grammar $\G_\NF(\mathcal{R_C})$  for instantiating variables.
In order to be able to apply these inferences,
\begin{RR}
according to the definition of term generation in Section~\ref{sb-production},
\end{RR}
we shall initiate the process by adding to the conjectures
one membership constraint for each variable.
\begin{definition} \label{def:decoration}
Let $C \cons[c]$ be a constrained clause such that 
$c$ contains no membership constraint.
The \emph{decoration} of $C\cons[c]$, 
denoted $\decorate(C \cons[c])$
is the set of clauses 
$C\cons[c \land x_1 \deriv \state{u_1}\land\ldots\land {x_n} \deriv \state{u_n}]$
where $\{ x_1, \ldots, x_n \} = \var(C)$,
and for all $i \in [1..n]$,  $\state{u_i} \in Q_\NF(\RC)$
and $\sort(u_i) = \sort(x_i)$.
\end{definition}
The definition of $\decorate$ is extended to set of constrained clauses as expected.
A constrained clause $C \cons[c]$ is said \emph{decorated} if
$c = d \land x_1 \deriv \state{u_1}\land\ldots\land {x_n} \deriv \state{u_n}$
where $\{ x_1, \ldots, x_n \} = \var(C)$,
and for all $i \in [1..n]$,  $\state{u_i} \in Q_\NF(\RC)$,
$\sort(u_i) = \sort(x_i)$, and $d$ does not contain membership constraints.

\textsf{Simplification},
resp. \textsf{Inductive Simplification},
reduces conjectures according 
to the rules of Section~\ref{sb-src},   
resp. \ref{sb-srf}.
%
%
\textsf{Inductive Narrowing} generates new subgoals by application of
the production rules of the constrained grammar $\G_\NF(\RC)$ 
until the obtained clause is deep enough to cover left-hand side of
rules of $\RD$. Each obtained clause must be simplified
by one the rules of Figure~\ref{srdf}
(otherwise, if one instance cannot be simplified, then the rule \textsf{Inductive Narrowing} cannot be applied).
For sake of efficiency, the application 
can be restricted to 
so called \emph{induction variables}, as defined in~\cite{Bouhoula97jsc}
\begin{RR}
(see Section~\ref{sec:ex-test-set})
\end{RR}
while preserving all the results of the next section.
\textsf{Narrowing} 
is similar and uses the rules of Figure~\ref{srcf} for simplification.
This rule permits to eliminate the ground reducible constructor terms 
in a clause by simplifying
their instances, while deriving conjectures considered as new subgoals.  
The criteria on depth 
is the same for \textsf{Inductive~Narrowing} and
\textsf{Narrowing} and is a bit rough, for sake of clarity of the inference rules. 
However, in practice, it can be replaced by a tighter condition
(with, \emph{e.g.}, a distinction between $\RC$ and $\RD$)
while preserving the results of the next section.
\textsf{Subsumption} deletes clauses redundant with
axioms of $\R$, induction hypotheses of $\Hyp$ and other
conjectures not yet proved (in $\E$). 


\begin{ABS}
\begin{example} \label{ex:proof}
Let us come back to the running example of sorted lists, 
with the constructor system $\RC$ containing (\ref{eq:ins-double}) and (\ref{eq:ins-ord})
and the defined system $\RD$ containing the axioms (\ref{eq:occur2-empty}-\ref{eq:occur2-geq}) given in introduction%
\footnote{In~(\ref{eq:occur2-leq}), the constraints $y_1 \approx \ins(x_2, y_2), y_1 \deriv \mathsf{NF}$
can be replaced by $y_1 \deriv \state{\ins(x_2, y_2)}$.}
%
together with the following axioms defining a variant $\occur$ for the membership:
\begin{align}
x \occur \emptyset  & \rightarrow \false
   \tag{$\mathsf{m}_0$} \label{eq:occur-empty}\\
x_1 \occur \ins(x_2, y)  & \rightarrow \true \cons[x_1 \approx x_2]
 \tag{$\mathsf{m}_1$} \label{eq:occur-eq} \\
x_1 \occur \ins(x_2, y)  & \rightarrow x_1 \occur y \cons[x_1 \not\approx x_2] 
  \tag{$\mathsf{m}_2$} \label{eq:occur-neq}
\end{align}
We show, using our procedure,
that the conjecture
$x \occurbis y  =  x \occur y$
is an inductive theorem of $\R$, i.e. that the two variants $\occurbis$ and $\occur$ of membership are equivalent.

\noindent The normal-form grammar $\G_\NF(\RC)$ is described in example~\ref{ex:grammar-nf}. The decoration of the conjecture with its non-terminal gives the two clauses:
$x \occurbis y = x \occur y  \cons[x \deriv \dxstate{x}{\Nat}, y \deriv \dxstate{x}{\Set} ]$
and
$x \occurbis y  = x \occur y  \cons[x \deriv \dxstate{x}{\Nat}, y \deriv \state{\ins(x_1, y_1)} ]$.

The application of the production rules of $\G_\NF(\RC)$ to the first of these clauses
(in \textsf{Narrowing}) gives:
$x \occurbis \emptyset =  x \occur \emptyset$ 
which is reduced, 
using~(\ref{eq:occur2-empty}) and~(\ref{eq:occur-empty}),
to the tautology $\false = \false$.
For the second clause, applying $\G_\NF(\RC)$ returns:
\begin{align}
%
x \occurbis \ins(x_1, \emptyset) & =  x \occur \ins(x_1, \emptyset) 
  \cons[x, x_1 \deriv \dxstate{x}{\Nat}]
  \label{ex2-ii}\\ 
x \occurbis \ins(x_1,  \ins(x_2, \emptyset)) & =  x \occur \ins(x_1,  \ins(x_2, \emptyset))
 \cons[x, x_1, x_2 \deriv \dxstate{x}{\Nat}, x_1 \prec x_2]  
 \label{ex2-iii} \\ 
x \occurbis \ins(x_1, \ins(x_2, y_2)) & =  x \occur \ins(x_1,  \ins(x_2, y_2)) \notag \\
 & \cons[x, x_1, x_2 \deriv \dxstate{x}{\Nat}, y_2 \deriv \state{\ins(x_3, y_3)}, x_1 \prec x_2, x_2 \prec x_3]  
 \label{ex2-iv} 
\end{align}
The subgoals (\ref{ex2-ii}) and  (\ref{ex2-iii}) 
can be simplified by \textsf{Rewrite Splitting} 
with~(\ref{eq:occur2-eq}),~(\ref{eq:occur2-leq}) and~(\ref{eq:occur2-geq})
into clauses reduced into tautologies 
(see~\cite{BouhoulaJacquemard08rrlsv07} for details).



\noindent 
The subgoal~(\ref{ex2-iv}) is implified by \textsf{Rewrite Splitting} 
with (\ref{eq:occur2-eq}-\ref{eq:occur2-geq}) into 3 clauses.
Let us consider the third one, obtained with (\ref{eq:occur2-geq}):
$x \occurbis \ins(x_2, y_2)  =  x \occur \ins(x_1,  \ins(x_2, y_2)) 
 \cons[x, x_1, x_2, x_3 \deriv \dxstate{x}{\Nat}, y_2 \deriv \state{\ins(x_3, y_3)}, x_1 \prec x_2, x_2 \prec x_3, x_1 \prec x]$.
%
%
It is simplified by \textsf{Inductive Rewriting} with (\ref{eq:occur-neq}) into:\\
$
x \occurbis \ins(x_2, y_2)  =  x \occur \ins(x_2, y_2)
  \cons[x, x_2, x_3 \deriv \dxstate{x}{\Nat}, y_2 \deriv \state{\ins(x_3, y_3)}, 
             x_2 \prec x_3 ] 
$.

\noindent         
At this point, we are allowed to use the conjecture $x \occurbis y  =  x \occur y$
as an induction hypothesis with \textsf{Inductive Rewriting}, 
it returns the tautology:
\[
 x \occurbis \ins(x_2, y_2)  =  x \occurbis \ins(x_2, y_2)
  \cons[x, x_2, x_3 \deriv \dxstate{x}{\Nat}, y_2 \deriv \state{\ins(x_3, y_3)}, 
             x_2 \prec x_3 ] 
\]
%
%
The ommited details in the proof of the conjecture can be found 
in~\cite{BouhoulaJacquemard08rrlsv07}.
Note that this proof does not require the manual addition of lemma.
\finex
\end{example}
\end{ABS}

\subsection{Soundness and Completeness} \label{sec-soundcomplete}
We show now that our inference system is sound and refutationally complete.
The proof of soundness is not straightforward. 
The main difficulty is to make sure that the exhaustve application of the 
rules preserve a counterexample when one exists. 
We will show more precisely that a \emph{minimal}
counterexample is preserved along a \emph{fair} derivation. 

A \emph{derivation} is a sequence of inference steps generated by a
pair of the form $(\E_0,\emptyset)$, using the inference rules in $\Ind$, written 
$(\E_0,\emptyset) \InferInd (\E_1,\Hyp_1) \InferInd \ldots$
It is called {\em fair} if the set of persistent constrained clauses 
$(\cup_i \cap_{j\geq i} \E_j)$ is
empty or equal to $\{\bot\}$. 
The derivation is said to be a \emph{disproof} in the latter case,
and a {\em success} in the former.

Finite success is obtained when the set of conjectures to be proved is
exhausted.  Infinite success is obtained when the procedure diverges,
assuming fairness.  When it happens, the clue is to guess some
lemmas which are used to subsume or simplify the generated infinite
family of subgoals, therefore stopping the divergence.  This is
possible in principle with our approach, 
since lemmas can be specified in the same way as axioms are.

\begin{theorem}[Soundness of successful derivations] \label{sosdt} \label{th:soundness}
Assume that $\RC$ is terminating and that $\R$ is sufficiently complete.
Let $\D_0$ be a set of unconstrained clauses
and let $\E_0 = \decorate(\D_0)$.
If there exists a successful derivation
$(\E_0, \emptyset)\InferInd (\E_1 , \Hyp_1)\InferInd~\cdots$ 
then $\R \modelsind \D_0$.
\end{theorem}
\begin{ABS}
\begin{proof}
(sketch, see~\cite{BouhoulaJacquemard08rrlsv07} for a complete proof).
The proof uses the fact that, under the hypotheses of Theorem~\ref{th:soundness},
$\R \modelsind \E_0$ implies $\R \modelsind \D_0$.

Intuitively, the reason is that 
in order to show that $\R \modelsind \D_0$, 
it is sufficient to show that $\R \models \D_0\sigma$ 
for all substitutions $\sigma$ whose images contain only ground constructor terms in normal form.
Every ground $\sigma$ can indeed be normalized into a substitution of this form
because $\RC$ is terminating and $\R$ sufficiently complete.
By definition of the decoration, the membership constraints and by construction of $\G_\NF(\RC)$,
this sufficient condition is a consequence of $\R \modelsind \E_0$.

We then show that $\R \modelsind \E_0$ by minimal counter-example.
Assume that $\R \not\modelsind \E_0$ and let $D_0$ be a clause, minimal wrt $\gg$, in the set:
\[ \bigl\{ D\sigma \bigm| D\cons[d] \in \cup_i \E_i,
\sigma \in \sol(d)~\mbox{\small is constructor and irreducible and }
 \R \not \models D\sigma \bigr\}.\]
Let $C \cons[c]$ be a clause of  $\cup_i \E_i$ minimal by subsumption ordering 
and $\theta \in \sol(c)$, irreducible and constructor ground substitution,
be such that $C\theta = D_0$.
%
We show in~\cite{BouhoulaJacquemard08rrlsv07} that whatever inference, other than \textsf{Disproof},
is applied to $C\cons[c]$, a contradiction is obtained, 
hence that the above derivation is not successful.\qed
\end{proof}
\end{ABS}
\begin{RR}
\begin{proof}
Assume that $\R \not\modelsind \D_0$,
and let $(\E_{0},\emptyset) \InferInd (\E_{1},\Hyp_{1}) \InferInd \cdots $  
be an arbitrary successful derivation.
By the following Fact, we have that $\R \not\modelsind \E_0$.
\begin{fact} \label{fac:soundness}
Assume that $\RC$ is terminating and that $\R$ is sufficiently complete.
If $\R \modelsind \E_0$ then $\R \modelsind \D_0$.
\end{fact}
\begin{proof}
Assume that $\R \modelsind \E_0$  and that 
for some clause $C \in \D_0$ we have $\R \not\modelsind C$.
Let $\{ C\cons[c_1], \ldots, C\cons[c_n] \} = \decorate(C)$.
For all $i \in [1.. n]$, we have $\R \modelsind C\cons[c_i]$, 
but there exists $\sigma \notin \cup_{i=1}^{n}\sol(c_i)$
such that $\R \not\models C\sigma$.
Since $\R$ is sufficiently  complete and $\RC$ is terminating, 
we can rewrite $\sigma$ into a constructor and $\RC$-irreducible
ground substitution $\sigma'$. 
By Lemma~\ref{lem:GNFRC}, it follows that 
$\sigma' \in \sol(c_i)$ for some $i \in [1..n]$, and therefore
that $\R \models C\sigma'$, a contradiction
with $\R \not\modelsind C\sigma$.~\hfill$\Box$
\end{proof}

Let $D_0$ be a clause, minimal wrt $\gg$, in the set:
\[ \bigl\{ D\sigma \bigm| D\cons[d] \in \cup_i \E_i,
\sigma \in \sol(d)~\mbox{\small is constructor and irreducible and }
 \R \not \models D\sigma \bigr\}\]
Note that such a clause exists
since we have proved that $\R \not\modelsind \E_0$.
Let $C \cons[c]$ be a clause of  $\cup_i \E_i$
minimal by subsumption ordering 
and $\theta \in \sol(c)$, 
irreducible and constructor ground substitution,
be such that $C\theta = D_0$.

We show that whatever inference, other than \textsf{Disproof},
is applied to $C\cons[c]$, a contradiction is obtained, 
hence that the above derivation is not successful.\\

\noindent \textsf{Inductive Narrowing}.
Suppose that the inference \textsf{Inductive Narrowing} is applied to $C\cons[c]$.  
By hypothesis, $C$ has been decorated, \emph{i.e.}
$c = d \land x_1 \deriv \state{u_1}\land\ldots\land {x_n} \deriv \state{u_n}$
with $\{ x_1, \ldots, x_n \} = \var(C)$ and for all $i \in [1..n]$,  $\state{u_i} \in Q_\NF(\RC)$.
Hence, since $\theta \in \sol(c)$, 
there exists $\sigma$ and $\tau$ such that $\theta = \sigma\tau$ and
$C\cons[c] \vdash^+ C\sigma\cons[c']$.\\

\noindent $C\cons[c] \sigma$ cannot be a tautology and $c$ cannot be unsatisfiable and
therefore the rule \textsf{Inductive Deletion} cannot be applied.\\

\noindent Let $C'$ be the result of the application 
of the rule \textsf{Inductive Rewriting} to $C\sigma\cons[c']$. 
The instances of clauses of $\Hyp \cup \E \cup \{C\}$ used in the
rewriting step are smaller than $C \theta$ wrt $\gg$, and
therefore, they are inductive theorems of $\R$. 
Hence $\R \not \models C' \tau$.
Moreover, $C\theta \gg C' \tau$ 
and $C' \in \cup_i \E_i$, which is a contradiction.\\

\noindent With similar arguments as above, we can show that the
rule \textsf{Inductive Contextual Rewriting} cannot be applied to  $C\sigma\cons[c']$.\\

\noindent Assume that the rule \textsf{Rewrite Splitting} is applied to
$C[t]_p\sigma\cons[c']$. Let
\[\{ \Gamma_1 \Rightarrow l_1 \rightarrow r_1 \cons[c_1],~\ldots~ ,\Gamma_n \Rightarrow
l_n \rightarrow r_n\cons[c_n]\}\]
be the non-empty subset of $\RD$ such that for all $i$ in $[1..n]$, $t=l_i\sigma_i$ and 
\[\R \modelsind \Gamma_1 \sigma_1\cons[c' \wedge c_1 \sigma_1] \vee
 \ldots \vee 
 \Gamma_n \sigma_n\cons[c' \wedge c_n \sigma_n]\]
The result of the application of \textsf{Rewrite~Splitting} is:
\[ \{\Gamma_1\sigma_1\Rightarrow C[r_1\sigma_1]_p \cons[c' \land c_1\sigma_1],
\ldots, \Gamma_n\sigma_n\Rightarrow C[r_n\sigma_n]_p \cons[c' \land  c_n\sigma_n] \}\]
Then there exists $k$ such that $\R \models \Gamma_k \sigma_k \delta$ 
for some $\delta \in Sol(c' \wedge  c_k \sigma_k)$.  
Let $C_k \equiv \Gamma_k\sigma_k\Rightarrow C[r_k\sigma_k]_p \cons[c' \land c_k\sigma_k]$,
we have $\R \not \models C_k \delta$, 
since $\R \models \Gamma_k\sigma_k\delta$, $\R \models t \delta = r_k \sigma_k \delta$, 
and $\R \not \models C \theta$.  On the other hand, $C\theta \gg C_k \delta$ 
since $\{ t \} >^\mul \Gamma_k \sigma_k$, and 
$t > r_k \sigma_k$. This contradicts the minimality of $C\theta$.\\

\noindent \textsf{Narrowing, Inductive Simplification and Simplification}.
These cases are similar to the previous one.\\

\noindent \textsf{Subsumption}:
Since $\R \not\models C \theta$, $C \cons[c]$ cannot be subsumed by an axiom of $\R$.
If there exists $C'\cons[c'] \in \Hyp \cup (\E \setminus \{C\cons[c]\})$ such that
$C\cons[c] \equiv C'\delta \cons[c'\delta] \vee D$, then we have 
$\R \not\models C' \delta \theta$ ($\theta \in \sol(c')$).
Hence, $r=\emptyset$ and $\delta = \emptyset$, since $C\cons[c]$ is minimum in
$\cup_i \E_i$ wrt subsumption ordering. Therefore, 
$C' \not\in(\E\setminus \{C\})$. 
Moreover, $C' \not \in \Hyp$, otherwise the
inference \textsf{Inductive Narrowing} or \textsf{Narrowing} could
also be applied to $C \cons[c]$, 
in contradiction with previous cases.  Hence, \textsf{Subsumption} 
cannot be applied to $C\cons[c]$. \hfill $\Box$\\
\end{proof}
\end{RR}

Since there are only two kinds of fair derivations, we obtain
as a corollary:
\begin{corollary}[Refutational completeness] \label{ref-comp} \label{cor:refutational-completeness}
Assume that $\RC$ is terminating and that $\R$ is sufficiently complete. 
Let $\D_0$ be a set of unconstrained clauses and let $\E_0 = \decorate(\D_0)$.
If $\R \not \modelsind \E_0$, then 
all fair derivations starting from $(\E_0,\emptyset)$ end up with $(\bot,\Hyp)$.
\end{corollary} 

When we assume 
that all the variables in goals are decorated 
(restricting the domain for this variables to ground constructor irreducible terms), 
the above hypotheses 
that $\RC$ is terminating and $\R$ is
sufficiently complete can be dropped.
\begin{theorem}[Soundness of successful derivations] \label{sosdt-constrained} \label{th:soundness-constrained}
Let $\E_0$ be a set of decorated constrained clauses.
If there exists a successful derivation
$(\E_0, \emptyset)\InferInd (\E_1 , \Hyp_1)\InferInd~\cdots$ 
then $\R \modelsind \E_0$.
\end{theorem}
\begin{ABS}
\begin{proof} (sketch).
We use the second part of the proof of Theorem~\ref{th:soundness}
(which does not use the sufficient completeness of $\R$ and termination of $\RC$).
With the hypothesis that the clauses of $\E_0$ are decorated, 
the fact given at the beginning of this proof is indeed no more needed ($\D_0 = \E_0$).
The restriction to substitutions into ground constructor normal forms
in order to show that $\R \not\modelsind \E_0$ is made explicit by the membership constraints
in the decoration.
\qed
\end{proof}
\end{ABS}
\begin{RR}
\begin{proof}
The proof is the same as for 
Theorem~\ref{th:soundness} except that we do no need
the Fact~\ref{fac:soundness} since the goals of $\E_0$ are
already decorated. Hence we do neither need the hypotheses that 
$\RC$ is terminating and that $\R$ is sufficiently complete
which where only used for the proof of Fact~\ref{fac:soundness}.
\hfill~$\Box$
\end{proof}
\end{RR}

\begin{RR}
As a consequence, of the above theorem,
we immediately have the refutational completeness
of our inference system 
if the goals are decorated constrained clauses.
\end{RR}

\begin{corollary}[Refutational completeness] \label{ref-comp-constrained}
Let $\E_0$ be a set of decorated constrained clauses.
If $\R \not \modelsind \E_0$, then 
all fair derivations starting from $(\E_0,\emptyset)$ end up with $(\bot,\Hyp)$.
\end{corollary} 
We shall see in Section~\ref{sec:partial} some example
of applications of Theorem~\ref{th:soundness-constrained}
and Corollary~\ref{ref-comp-constrained} to specifications
which are not sufficiently complete.

\begin{ABS}
\end{ABS}
\begin{RR}
\bigskip
Our inference system can refute false conjectures. 
This result is a
consequence of the following lemma.

\begin{lemma}
\label{ref-sound}
let $(\E_i, \Hyp_i)$ $\InferInd$
$(\E_{i+1},$ $\Hyp_{i+1})$ be a derivation step.  
If $\R\modelsind \E_i \cup \Hyp_i$ then $\R
\modelsind  \E_{i+1} \cup \Hyp_{i+1}$.
\end{lemma}

\begin{proof}
Let $C\cons[c]$ be a clause in $\E_i$ and $(\E_i \cup \{C\cons[c]\},\Hyp_i) \InferInd
(\E_{i+1},\Hyp_{i+1})$ be a derivation step obtained by the application 
of an inference to $C\cons[c]$ and assume that $\R \modelsind
\E_i \cup \Hyp_i$.
By hypothesis, the instances of clauses of $\Hyp \cup \E \cup \{C\cons[c]\}$
which are used during rewriting steps, are valid.
Hence, we can show that $\R \modelsind \E_{i+1} \cup \Hyp_{i+1}$   
by a case analysis according to the rule applied to $C\cons[c]$.~\cqfd
\end{proof}

The following lemma is also used in the proof of soundness of disproof.
\begin{lemma} \label{lem:R-to-RC}
If $\R$ is ground confluent and sufficiently complete
then for every constructor clause $C\cons[c]$,
if $\R \modelsind C\cons[c]$ then $\RC \modelsind C\cons[c]$.
\end{lemma}
\begin{proof}
Let $\tau \in \sol(c)$ be a substitution grounding for $C$.
By the sufficient completeness of $\R$, we may assume without loss of generality
that $\tau$ is a constructor substitution.
By hypothesis, $\R \models C\tau$.
Assume that for some literal $u = v$ of $C$, 
we have $\R \models u\tau = v\tau$.
Since $\R$ is ground confluent, 
it means that $u\tau \downarrow_\R v \tau$, 
and hence that $u\tau \downarrow_\RC v \tau$,
\emph{i.e.} $\RC \models u\tau = v\tau$,
because $u\tau, v\tau \in \T(\C)$.
Moreover, if $\R \models u\tau \neq v\tau$ 
then $\RC \models u\tau \neq v\tau$
because $\RC \subseteq \R$.
\hfill $\Box$
\end{proof}
\end{RR}

\begin{theorem}[Soundness of disproof] \label{Refutational-soundness} \label{th:refutational-soundness}
Assume that $\R$ is strongly complete and ground confluent.
If a derivation starting from $(\E_0,\emptyset)$ returns
the pair $(\bot,\Hyp)$, then $\R \not \modelsind \E_0$.
\end{theorem}
\begin{ABS}
\end{ABS}
\begin{RR}
\begin{proof} Under our assumptions, there exists a step $k$ in the
derivation, such that \textsf{Disproof} applies to a constrained clause 
$C\cons[c]$ in $\E_k$.

We prove first that $C\cons[c]$ is a constructor clause.
Assume indeed that $C\cons[c]$ contains a term of the form $f(t_1, \ldots ,t_n)$,
where $f \in \D$ and for all $i \in [1.. n]$, $t_i \in T(\C,\X)$.  
The constraint $c$ is satisfiable, 
otherwise \textsf{Inductive~Deletion} could be applied.
Let $\tau \in \sol(c)$.
Hence by Lemma~\ref{lem:GNFRC},
for each $x \in \var(C)$, $x\tau$ is in $\RC$-normal form.
We have now two possibilities:
\begin{enumerate}
\item for one $i \in [1..n]$, $t_i \tau$ is reducible.
  In this case, there exists a substitution $\sigma$
  such that $\tau = \sigma\theta$ and 
  $t_i \cons[c] \vdash^+ t_i \sigma \cons[c']$
  and $t_i\sigma$ contains as a subterm an instance of a left-hand side
  of rule of $\RC$. Therefore, either \textsf{Rewriting} or
  \textsf{Partial~Splitting} can be applied to $t_i\sigma \cons[c']$.
  It implies that \textsf{Narrowing} can be applied to $C\cons[c]$, 
  which is a contradiction.
\item
  every $t_i \tau$ is irreducible.
  The term $f(t_1, \ldots, t_n)\tau$ is reducible
  at root position because $f$ is strongly complete wrt $\R$.
  Then there exists $\sigma$ such that $\tau = \sigma\theta$ and 
  $f(t_1,\ldots, t_n) \cons[c] \vdash^+  f(t_1,\ldots, t_n)\sigma \cons[c']$
  and moreover $f(t_1,\ldots, t_n)\sigma$ is an instance of a left-hand side
  of rule of $\RD$. 
  Therefore, either \textsf{Inductive rewriting} or \textsf{Rewrite Splitting} 
  can be applied.
  Indeed the application condition of the latter inference
  is a consequence of the strongly completeness of $\R$. 
  Hence, the inference \textsf{Inductive Narrowing} can be applied to $C\cons[c]$, which  
  is a contradiction.  
\end{enumerate}
In conclusion, the clause $C\cons[c]$ contains only constructor terms.

Then, we deduce that $C\cons[c]$ contains ground irreducible terms only,
otherwise \textsf{Narrowing} would apply. 
Since \textsf{Validity} does not apply either, $C\cons[c]$ is not an inductive consequence of $\RC$.
By lemma~\ref{lem:R-to-RC}, and since $\R$ is ground confluent, we conclude 
that $C\cons[c]$ is not an inductive theorem of $\R$.
As a consequence, $\R \not\modelsind \E_k$.
Finally, by lemma~\ref{ref-sound}, we deduce that $\R\not\modelsind \E_0$.~\cqfd
\end{proof}
\end{RR}

\subsection{Handling Non-Terminating Constructor Systems} \label{se-ordering}
Our procedure applies rules of $\RC$ and $\RD$ only when they
reduce the terms wrt the given simplification ordering $>$.
This is ensured when the rewrite relation induced by $\RC$ and $\RD$
is compatible with $>$, and hence that $\RC$ and $\RD$ are
terminating 
\begin{ABS}
(separately).
\end{ABS}
\begin{RR}
(separately), like in the example of Section~\ref{se-example}.
\end{RR}
Note that this is in contrast with other procedures 
like~\cite{BouhoulaRusinowitch95jar,Bouhoula97jsc}
where the termination of the whole system $\R$ is required.

If $\RC$ 
is non-terminating then one can
apply e.g. the constrained completion technique~\cite{KirchnerKirchnerRusinowitch90ria}
in order to generate an equivalent orientable theory (with ordering constraints).
The theory obtained (if the completion succeeds)
can then be handled by our approach.  

\begin{example}\label{ex-order-completion}
Consider this non-terminating system for sets:
\begin{ABS}
\(
\{
\ins(x, \ins(x,y)) =  \ins(x,y), 
\ins(x, \ins(x',y)) = \ins(x',\ins(x,y))
\}
\).
\end{ABS}
\begin{RR}
\[
\begin{array}{rcl}
\ins(x, \ins(x,y)) & = & \ins(x,y)\\
\ins(x, \ins(x',y)) & = & \ins(x',\ins(x,y))
\end{array}
\]
\end{RR}
Applying the completion procedure we obtain the 
\begin{ABS}
constrained rules (\ref{eq:ins-double}) and (\ref{eq:ins-ord}).\finex
\end{ABS}
\begin{RR}
constrained system of Section~\ref{sec:example}.\finex
\end{RR}
\end{example}

\begin{ABS}
\subsection{Decision Procedures for Conditions in Inferences} \label{sb-scheme}
Constrained tree grammars are involved in the inferences
\textsf{Narrowing} and \textsf{Inductive Narrowing} in order to 
generate subgoals from goals,
by instantiation using the productions rules.
They are also the key for the decision procedures applied in order to check
the conditions of constraint unsatisfiability
(in rules for rewriting and
 \textsf{Inductive~Deletion}, \textsf{Deletion}, \textsf{Subsumption}),
ground irreducibility
and
validity of ground irreducible constructor clauses
(in the rules \textsf{Validity}, hence \textsf{Simplification}, and \textsf{Disproof}).
These conditions are decided by reduction into the 
decision problem 
of emptiness 
(of $L(\G,\state{u})$)
for constrained tree grammars build from $\G_\NF(\RC)$.
The decision rely on similar decision results for
constrained tree automata, some cases are summarized in~\cite{tata}.
The reductions are detailed in~\cite{BouhoulaJacquemard08rrlsv07}.
\end{ABS}

\begin{RR}

\section{Decision Procedures for Conditions in Inference Rules} \label{sb-scheme}
We present a reduction of the conditions in 
the inference rules of Figures~\ref{fi-defined}, \ref{fi-constructors}, and~\ref{fi-inference}
to emptiness decision problems for tree automata with constraints.
We deduce a decision procedure for these tests
in the case where the constraints in the specification are limited to
syntactic equality and disequality.

We assume here that, like in Theorem~\ref{th:soundness},
the inference system is applied to a set $\decorate(\D_0)$
where $\D_0$ is a set of unconstrained clauses.

\subsection{Reductions} \label{sb-decision}\label{sb-emptiness} 
Consider the following decision problems,
given two constrained grammars $\G$, $\G'$ and
two non terminals $\state{u}$, $\state{u'}$ of respectively  $\G$ and $\G'$,
\begin{description}
\item[(ED)] emptiness decision: $L(\G,\state{u}) = \emptyset$?
\item[(EI)] emptiness of intersection: $L(\G,\state{u}) \cap L(\G',\state{u'}) = \emptyset$?
\end{description}

\paragraph{\bf Ground instances.} \label{sb-instance}
Let $t \cons[c]$ be a constrained term (or clause) 
such that the constraint $c$ has the form 
\( x_1\deriv\state{u_1} \land \ldots \land x_m\deriv\state{u_m} \mathbin{\land} d \) 
where $d$ contains no membership constraints.
Note that starting with decorated clauses, 
any goal or subgoal occurring during the inference 
is of the above form.
%
The set of ground instances of $t$ satisfying $c$ is recognized 
by a constrained grammar $\G(t\cons[c])= \bigl(Q(t\cons[c]), \Delta(t\cons[c])\bigr)$
whose construction is described in Figure~\ref{instdgf}.

\noindent For technical reasons concerning non-terminals separation, 
we use in the construction of $\G(t\cons[c])$
a relabeling isomorphism ${}^\circ$ from the signature $(\mathcal{S},\F)$
to the signature $(\mathcal{S}^\circ,\F^\circ)$, 
such that the function symbol $f^\circ$ has profile $S_1^\circ \times \ldots \times S_n^\circ \to S^\circ$
if $f$ has profile $S_1 \times \ldots \times S_n \to S$, 
and its extension from $T(\F, \X)$ to $\T(\F^\circ, X)$, 
such that (recursively) $f(t_1,\ldots,t_n) = f^\circ(t_1^\circ,\ldots,t_n^\circ)$,
and for each $x\in \X$, $x^\circ = x$.
\begin{mafigure}[ht]
\begin{tabular}{l}
\quad \( Q(t\cons[\mathop{\bigwedge}_{i=1}^{m} x_i\deriv\state{u_i}\land d]) =  Q_\NF(\RC) \cup
 \{ \state{u^\circ} \mid u \unlhd t \}\)\\     
~\\ 
\quad $\Delta(t\cons[\mathop{\bigwedge}_{i=1}^{m} x_i\deriv\state{u_i}\land d])$ 
contains all the production rules of $\Delta_\NF(\RC)$ plus:~\\     
\( 
\begin{array}{rcl}  
\state{t^\circ} & := & g(\state{t_1^\circ},\ldots,\state{t_m^\circ}) \cons[d], \mbox{~if~} t = g(t_1,\ldots,t_m)\\           
\mbox{and every~} \state{f^\circ(v_1^\circ,\ldots,v_n^\circ)} & := & 
 f(\state{s_1},\ldots,\state{s_n}) \cons \mbox{~such that~} f(u_1,\ldots,u_n) \lhd t,\\ 
\end{array} 
\)\\
         \quad and  $\forall j \leq m$  if $v_j^\circ = x_i$ for some $i$, then $\state{s_j} = \state{u_i}$\\    
\phantom{\quad and  $\forall j \leq m$}
 if $v_j^\circ \in \X \setminus \{ x_1,\ldots,x_m\}$ then $\state{s_j} \in Q_\NF(\RC)$\\   
\phantom{\quad and  $\forall j \leq m$}
 if $v_j^\circ \notin  \X$ then $\state{s_j} = \state{v_j^\circ}$     
\end{tabular}   
\caption{Constrained Grammar $\G(t, c)$ Ground instances  \label{instdgf}}    
\end{mafigure}
\begin{lemma} \label{le-instance}
$L(\G(t\cons[c]), \state{t}) = \{ t\sigma \mid \sigma|_{\var(c)} \in\sol(c) \}$.
\end{lemma}
\begin{proof}
The proofs of both directions $\subseteq$ are straightforward
inductions
resp. on the length of a derivation of a term of $L(\G(t\cons[c]), \state{t})$
and on 
a ground instance $t\sigma$ such that $\sigma|_{\var(c)}$ is a solution of $c$.~\cqfd
\end{proof}

\paragraph{\bf Constraints unsatisfiability.}
This property is required for rules 
\textsf{Inductive~Rewriting}, \textsf{Inductive~Contextual~Rewriting},
\textsf{Rewrite~Splitting}, \textsf{Inductive~Deletion}, \textsf{Deletion},
and \textsf{Subsumption}.
\begin{lemma}
Given a constraint $c$, there exists
a constrained grammar $\G(c)$ such that 
$c$ is unsatisfiable iff  $L\bigl(\G(c)\bigr) = \emptyset$. 
\end{lemma}
\begin{proof}
Let $x_1,\ldots,x_m$ be the list of all the variables
occurring in $c$, eventually with repetition in case of multiple occurrences.
Let $y_1,\ldots,y_m$ be a list of fresh distinct variables, 
let $f^m$ be a new function symbol of arity $m$ and let $\tilde{c} = \bigwedge_{i=1}^{m} y_i \approx x_i$.
The constrained grammar $\G(c)$ is defined by
$\G(c) = \G\bigl(f^m(y_1,\ldots,y_m) \cons[c \land \tilde{c}]\bigr)$.~\cqfd
\end{proof}
\begin{corollary} \label{co-satifiability}
Constraints unsatisfiability is reducible to \emph{(ED)}.
\end{corollary}

\paragraph{\bf Ground (ir)reducibility.}
The rules \textsf{Validity}, hence \textsf{Simplification}, and \textsf{Disproof}
(by negation) check ground irreducibility.
\begin{lemma} \label{le-gi}
Ground reducibility and ground irreducibility decision are reducible to \emph{(EI)}.
\end{lemma}
\begin{proof}
By definition and Lemmas~\ref{lem:GNFRC} and~\ref{le-instance},
a constrained clause $C\cons[c]$ is ground reducible iff 
$L\bigl(\G(C\cons[c])\bigr) \cap L\bigl(\G_\NF(\RC), Q_\NF(\RC) \setminus \{ \xstate{x}{\re} \}\bigr) =\nolinebreak \emptyset$
and ground irreducible iff 
$L\bigl(\G(C\cons[c])\bigr) \cap L\bigl(\G_\NF(\RC), \xstate{x}{\re} \bigr) =\nolinebreak \emptyset$.
\cqfd
\end{proof}

\paragraph{\bf Validity of ground irreducible constructor clauses.} 
The rule \textsf{Validity}, hence \textsf{Simplification}, checks this property.
\begin{lemma} \label{le-valid}
When $\R$ is ground confluent, validity of ground irreducible constructor constrained clauses is reducible to \emph{(ED)}.
\end{lemma}
%
\begin{proof}
Let $C\cons[c]$ be a ground irreducible constructor constrained clause.
Let $\tilde{C}$ be the constraint obtained from $C$ by 
replacement of every equation $s=t$ (resp. disequation $s \neq t$) 
by the atom $s \approx t$ (resp. $s \not\approx t$).
Since $C\cons[c]$ is ground irreducible and  $\R$ is ground-confluent, we have that
$C\cons[c]$ is valid in the initial model of $\R$ 
iff every substitution $\sigma \in \sol(c)$ grounding for $C$  
is such that 
$\sigma \in \sol\bigl(\tilde{C}\bigr)$.
This is equivalent to $L\bigl(\G(C\cons[c \land \lnot\tilde{C}])\bigr) = \emptyset$.\cqfd
\end{proof}

\subsection{Decision} \label{sb-ta} 
It remains to give decision procedures for (ED) and (EI).
We proceed by reduction to analogous problems on
tree automata with (dis)equality constraints~\cite{tata},
for a class of tree grammars defined as follows.
\begin{definition}
A constrained grammar $\G$ is called \emph{normalized} if
for each of its productions 
\(\state{t} := f(\state{u_1},\ldots,\state{u_n}) \cons[c]\)
all the atomic constraints in $c$ have the
form $P(s_1,\ldots,s_k)$ where $P \in \mathcal{L}$
and $s_1,\ldots,s_k$ are strict subterms of $f(u_1,\ldots,u_n)$.
\end{definition}
%
Every normalized constrained grammar which contains only 
constraints with $\approx$, $\not\approx$ in its production rules is equivalent to a tree automaton
with equality and disequality constraints (AWEDC), see~\cite{tata} for a survey.
Therefore, constrained grammars inherit the properties of AWEDC concerning emptiness decision,
and (ED), (EI) are decidable for a normalized constrained grammar when
for each production \(\state{t} := f(\state{u_1},\ldots,\state{u_n}) \cons[c]\):
\begin{enumerate}
\item \label{en-bogaert}%
  the constraints in $c$ have the form $u_i \approx u_j$ or $u_i \not\approx u_j$~\cite{bogaert92stacs},
\item \label{en-comon}%
  the constraints in $c$ are only disequalities $s_1 \not\approx s_2$~\cite{ComonJacquemard03},
\item \label{en-dauchet}%
  the constraints in $c$ are equalities and disequalities, and for every (ground) constrained term $t\cons[c]$ 
  generated by $\G$, for every path $p \in \Pos(t)$, the number of
  subterms $s$  occurring along $p$ in $t$ and such that
   $s \approx s'$ or $s'\approx s$ is an atomic constraint of $c$ is
   bounded (independently from $t$  and $c$)~\cite{dauchet95jsc},
\item \label{en-caron}%
  the constraints in $c$ are equalities and disequalities,
  and for every (ground) constrained term $t\cons[c]$ generated by $\G$,
  for every path $p \in \Pos(t)$, the number of subterms $s$ satisfying the following conditions (i--iii)
  is bounded (independently from $t$ and $c$)~\cite{CaronComonCoquideDauchetJacquemard94}
  \begin{itemize}
  \item[(i)] $s$ occurs along $p$ in $t$, 
  \item[(ii)] $s \approx s'$ or $s'\approx s$ is an atomic constraint of $c$,
  \item[(iii)] $s$, $s'$ are not brothers in a subterm
    $f(\ldots,s,\ldots,s',\ldots)$ occurring on $p$.
  \end{itemize}
\end{enumerate}

\label{sb-restriction-RC}
\begin{theorem} \label{th-effective}
All the conditions of the simplification rules in Figures~\ref{fi-defined},\ref{fi-constructor}
and the inference rules in Figure~\ref{fi-inference} 
are decidable or make recursive call to the procedure itself
when $\R$ is ground confluent 
and, for all  $l \to r \cons[c] \in \RC$,
for all $s \approx s' \in c$, 
(resp. all $s \not\approx s' \in c$)
$s$ and $s'$ are either variables or strict subterms of $l$
(resp. variables or strict subterms occurring at sibling positions in $l$).
\end{theorem}
\begin{proof}
When the constraints of $\RC$ fulfill the above conditions, 
then $\G_\NF(\RC)$ is in category ~\ref{en-caron}, hence (ED) and (EI) are decidable.
Hence the conditions in the inference and simplification rules 
in Figures~\ref{fi-defined},\ref{fi-constructor},\ref{fi-inference}
which are not  recursive call,
are decidable by Corollary~\ref{co-satifiability} and Lemmas~\ref{le-gi},\ref{le-valid}.
\cqfd
\end{proof}

The algorithms provided in the literature for the emptiness decision 
for the classes \ref{en-bogaert}
to \ref{en-caron} of tree automata with equality and disequality constraints 
are all very costly, 
due to the inherent  complexity of the problem.
For instance, for the ``easiest'' class~\ref{en-bogaert}, 
the problem is EXPTIME-complete~\cite{tata}, 
see also~\cite{Kapur91acta,ComonJacquemard03} concerning class~\ref{en-comon}.
The problem is however less difficult for deterministic automata
(\textit{e.g.}, PTIME for class~\ref{en-bogaert}), 
like the one of Figure~\ref{fi-nf}.

Cleaning algorithms, which may behave better in the average, 
have been proposed~\cite{CaronComonCoquideDauchetJacquemard94} 
for optimizing emptiness decision.
An interesting aspect of the cleaning algorithm is its monotonicity:
an incremental change on the automaton in input causes only an incremental change of the 
intermediate structure constructed by the algorithm for emptiness decision.
This should permit to reuse such structures in our setting because all the constrained grammars
of Section~\ref{sb-emptiness} are incrementally obtained from the
unique normal form grammar  $\G_\NF(\RC)$.

Another promising approach for implementation is 
the use of first-order saturation techniques.
It has been studied for solving various decision problem 
for several classes of tree automata
with or without constraints~\cite{JRV-jlap08,CJP-fossacs07,JGL-ipl2005}.
\end{RR}

\section{Handling Partial Specifications} \label{se-partial} \label{se-powerlist} \label{sec:powerlist} \label{sec:partial}
The example of sorted lists 
\begin{ABS}
(Example~\ref{ex:proof})
\end{ABS}
\begin{RR}
in Section~\ref{sec:example} 
\end{RR}
can be treated with our procedure because it 
is based on a sufficiently complete
and ground confluent 
conditional constrained TRS $\R$ whose constructor part $\RC$ is terminating.
Indeed, under these hypotheses,
Theorem~\ref{th:soundness} ensures the soundness
of our procedure for proving inductive conjectures on this specification,
and Corollary~\ref{cor:refutational-completeness}
and Theorem~\ref{th:refutational-soundness}
ensure respectively refutational completeness and soundness of disproof.

For sound proofs of inductive theorems
wrt specifications which are not sufficiently complete,
%
we can rely on Theorem~\ref{th:soundness-constrained} and Corollary~\ref{ref-comp-constrained} 
which 
do not require sufficient completeness of the specification 
but instead suppose that the conjecture is decorated,
i.e. that each of its variables is constrained to belong
to a language associated to a non-terminal of the normal-form
(constrained) grammar.
In this section, we propose two applications of this principle
of decoration of conjectures to the treatment of partial specifications.
\begin{RR}
We treat the case where the specification of defined function is partial in Section~\ref{sec:partial-defined},
and the case where axioms for constructors are partial in Section~\ref{sec:partial-constructor}.
\end{RR}

\begin{ABS}
\paragraph{\bf Partially Defined Functions.}
An inductive proof of a decorated conjecture $C$ in $\R$
remains valid in an extension of $\R$ (possibly not complete).
\end{ABS}

\begin{RR}
\subsection{Partially Defined Functions}  \label{sec:partial-defined}
Under the condition that the conjecture is decorated, 
extending a given sufficiently complete 
specification with additional axioms for defining partial
(defined) functions preserves successful derivations.
\end{RR}
\begin{theorem} \label{th:partial}
Assume that $\R$ is sufficiently complete and
let $\R'$ be an consistent extension 
of $\R$ where $\RC' = \RC$ and 
$\RD' = \RD \cup \RD''$ 
($\RD''$ defines additional partial defined functions).
Let $\E_0$ be a set of decorated constrained clauses.
Every derivation 
$(\E_0, \emptyset)\InferInd \cdots$ 
successful wrt $\R$ is also a successful derivation wrt $\R'$.
\end{theorem}

\begin{ABS}
\noindent In~\cite{BouhoulaJacquemard08rrlsv07}, we use 
Theorem~\ref{th:partial}
for the proof of conjectures on
an extension of the specification of Example~\ref{ex:proof}
with the incomplete definition of a function $\minl $.
\end{ABS}

\begin{RR}
\begin{proof}
The grammars $\G_\NF(\RC')$ and $\G_\NF(\RC)$ are the same.
Therefore every inference step wrt $\R$ is also 
an inference step wrt $\R'$.\cqfd
\end{proof}
We apply Theorem~\ref{th:partial} to a partial extension of 
the specification of Section~\ref{sec:example}.
%
\paragraph{\bf Specification of min for sorted lists.}
Let us complete the specification of Section~\ref{sec:example}
with a new defined symbol $\minl: \Set \to \Nat$
and the following rules of $\R_\D$:
\[  
\begin{array}{rcl}
\minl(\ins(x,\emptyset)) & \to & x\\
\minl(\ins(x, \ins(y,z))) & \to & \minl(\ins(x,z)) \cons[x \prec y]\\[2mm]
\end{array} 
\]
The function $\minl$ is not sufficiently complete wrt $\R$
(the case $\minl(\emptyset)$ is missing).
\paragraph{\bf Proof of two conjectures for min.}
We shall prove, using our inference system,
that the two following constrained and decorated conjectures are inductive
theorems of $\R$.
\begin{equation} \label{ex-goal2}
\minl(\ins(x, \ins(y,z))) \to \minl(\ins(y,z)) 
\cons[x \succcurlyeq y \wedge x,y \deriv \dxstate{x}{\Nat} \wedge z \deriv \dxstate{x}{\Set}] 
\end{equation}
\begin{equation} \label{ex-goal3}
\minl(\ins(x, \ins(y,z))) \to \minl(\ins(y,z)) 
\cons[x \succcurlyeq y \wedge x,y \deriv \dxstate{x}{\Nat}
\wedge z \deriv \state{\ins(x_1,x_2)} ] 
\end{equation}


Let us now prove that the conjecture (\ref{ex-goal2}) is an inductive
theorem of $\R$.
We start by the simplification of (\ref{ex-goal2}) using a
\textsf{Partial~Splitting}.  
We obtain:
\begin{equation} \label{}
\minl(\ins(y,z)) =  \minl(\ins(y,z)) 
\cons[x \approx y \wedge x \succcurlyeq y \wedge  
x,y \deriv \dxstate{x}{\Nat} \wedge z \deriv \dxstate{x}{\Set} ] 
            \label{ex-goal2.1}
\end{equation}
\begin{equation} \label{}
\minl(\ins(x,\ins(y,z))) =  \minl(\ins(y,z)) 
\cons[x \not \approx y \wedge  x \succcurlyeq y \wedge  
x,y \deriv \dxstate{x}{\Nat} \wedge z \deriv \dxstate{x}{\Set} ] 
             \label{ex-goal2.2}
\end{equation}

The clause (\ref{ex-goal2.1}) is a tautology.
Subgoal (\ref{ex-goal2.2}) is simplified using 
\textsf{Partial~Splitting} again.  We obtain:
\begin{multline}
\minl(\ins(y,\ins(x,z))) =  \minl(\ins(y,z))\\
\cons[x \succ y \wedge x \succcurlyeq y \wedge x \not \approx y \wedge
  x,y \deriv \dxstate{x}{\Nat} \wedge z \deriv \dxstate{x}{\Set} ] 
            \label{ex-goal2.2.1}
\end{multline}
\begin{multline} 
\minl(\ins(y,\ins(x,z))) =  \minl(\ins(y,z))\\
\cons[x \nsucc y \wedge x \succcurlyeq y \wedge x \not \approx y \wedge
  x,y \deriv \dxstate{x}{\Nat} \wedge z \deriv \dxstate{x}{\Set} ] 
            \label{ex-goal2.2.2}
\end{multline}
\noindent Subgoal (\ref{ex-goal2.2.1}) is simplified by $\RD$ into $\minl(\ins(y,z)) = \minl(\ins(y,z))$, a tautology.
Subgoal (\ref{ex-goal2.2.2}) can also be deleted since the constraint 
$x \nsucc y, x \succcurlyeq y, x \not \approx y$ is unsatisfiable
This ends the proof that (\ref{ex-goal2}) is an inductive theorem of
$\R$.

The proof of (\ref{ex-goal3}) follows the same steps.

Note that by Theorem~\ref{th:partial} the proofs of the decorated conjectures
(\ref{ex-decore1}), (\ref{ex-decore2}) 
and (\ref{ex2-deco1}), (\ref{ex2-deco2})
in Section~\ref{sec:example} remain valid for
the above extended specification. 
\end{RR}

\begin{ABS}
\paragraph{\bf Partial Constructors.}
\end{ABS}
\begin{RR}
\subsection{Partial Constructors and Powerlists} \label{sec:partial-constructor}
\end{RR}
The restriction to decorated conjectures also permits
to deal with partial constructor functions.
In this case,  we are generally interested in proving 
conjectures only for constructor terms in the definition domain of
the defined function (well-formed terms).
\begin{ABS}

In~\cite{BouhoulaJacquemard08rrlsv07}, we
present an example of automatic proof where
\end{ABS}
\begin{RR}
This is possible with our procedure when
\end{RR}
$\RC$ is such that 
the set of well-formed terms 
is the set of constructor $\RC$-normal forms.
Hence, decorating the conjecture 
with grammar's non-terminals, as in Theorem~\ref{th:soundness-constrained},
amounts in this case at
restricting the variables to be instantiated by well-formed terms.

\begin{ABS}
The example is a specification of powerlists
(lists of $2^n$ integers stored in the leaves of a complete binary tree)
also treated in~\cite{Kapur}.
A particularity of this example is that $\RC$
contains constraints of the form $t \sim t'$
meaning that $t$ and $t'$ are well-formed lists 
of the same length
(\textit{i.e.} balanced trees of the same depth).
Such constraints are added to $\G_\NF(\RC)$
and we show that emptiness is decidable for these grammars
by reduction to the same problem for visibly tree automata
with one memory~\cite{CJP-fossacs07}.
\end{ABS}

\begin{RR}
We illustrate this approach in this section
with an example of application of
Theorem~\ref{th:soundness-constrained} to a 
non complete specification of powerlists.

\paragraph{\bf Specification of powerlists.}
A powerlist~\cite{Misra94powerlist} is a list of length $2^n$ (for $n \geq 0$)
whose elements are stored in the leaves of a balanced binary tree.
Kapur gives in~\cite{Kapur} a specification of powerlists
and some proofs of conjectures with an extension of RRL mentioned in introduction.
This example is carried out with an extension of the algebraic specification approach
where some partial constructor symbols are restricted by 
\emph{application conditions}.
We propose below another specification of powerlists which contains only constrained 
rewrite rules, and which can be efficiently handled by our method.

We consider a signature for representing powerlists of natural numbers,
with the sorts:
$\mathcal{S} = \{ \Nat, \List \}$ and the constructor symbols:
\[
\C = \bigl\{
0 : \Nat, s : \Nat \to \Nat,
v : \Nat \to \List, 
\tie: \List \to \List,
\bot: \List
\bigr\} \]
The symbols $0$ and $s$ are used to represent the natural numbers in 
unary notation, $v$ creates a singleton powerlist $v(n)$ of length 1 from a number $n$,
and $\tie$ is the concatenation of powerlists.
The operator $\tie$ is restricted to well balanced constructor terms of the same depth.
In order to express this property, we shall consider a constructor rewrite system 
$\RC$ which reduces  to $\bot$ every term $\tie(s, t)$ which is not well balanced.
This way, only the well defined powerlists are $\RC$-irreducible.
For this purpose, we shall use a new binary constraint predicate $\sim$ 
defined on constructor terms of sort $\List$ as the smallest equivalence such that:
\[
\begin{array}{rcl}
v(x) & \sim & v(y) \quad\mbox{for~all~}x ,y : \Nat\\
\tie(x_1, x_2) & \sim & \tie(y_1, y_2) 
	\quad\mbox{iff~}x_1 \sim x_2 \sim y_1  \sim y_2	
\end{array}
\]

The constructor TRS $\RC$ 	has one rule constrained by $\sim$:
\[ 
\begin{array}{c}
\tie(y_1, y_2) \to \bot \cons[y_1 \not\sim y_2] \quad
\tie(\bot, y) \to \bot \quad 
\tie(y, \bot) \to \bot
\end{array}
\]

\paragraph{\bf Tree grammars with $\sim$-constraints on brother subterms.}
The normal form tree grammar $\G_\NF(\RC)$ associated to $\RC$
generates the well founded ground constructor terms.
Its non-terminals, according to the construction in Section~\ref{sb-nf},
are: $\xstate{x}{\Nat}$, $\xstate{x}{\List}$,
$\state{\bot}$, $\state{\tie(x_1, x_2)}$
and its production rules:
\[
\begin{array}{rcll}
\multicolumn{4}{l}{%
\xstate{x}{\Nat} := 0 \quad
\xstate{x}{\Nat} := s(\xstate{x_2}{\Nat}) \quad
\xstate{x}{\List} := v(\xstate{x}{\Nat}) \quad
\state{\bot} := \bot
} 
\\[2mm]
\state{\tie(x_1, x_2)} & := &
 \tie(\xstate{x_3}{\List}, \xstate{x_4}{\List}) & \cons[x_3^\List \sim x_4^\List]\\
\state{\tie(x_1, x_2)} & := & 
 \tie\bigl(\state{\tie(x_3, x_4)}, \state{\tie(x_5, x_6)}\bigr) & \cons[ \tie(x_3, x_4) \sim \tie(x_5, x_6)]\\
\end{array}
\]
Note that all the constraints in these production rules
are applied to brother subterms.
We have omitted in the above list the non-terminal $\xstate{x}{\re}$,
and production rules of the form:
$\xstate{x}{\re} := \tie(\xstate{x_1}{\List}, \xstate{x_2}{\List})\cons[x_1^\List \not\sim x_2^\List]$ or
$\xstate{x}{\re} := \tie(\state{\bot}, \xstate{x_2}{\List})$.

The emptiness problem is decidable for such constrained tree grammars.
This can be shown with an adaptation of the proof in~\cite{bogaert92stacs}
to $\sim$-constraints (instead of equality constraints) or also 
by an encoding 
into the visibly tree automata with one memory of~\cite{CJP-fossacs07}.

\paragraph{\bf Proof of a conjecture.}
We add to the specification a defined symbol $\rev$:
$\D = \bigl\{ \rev: \List \to \List \bigr\}$
and a defined TRS $\RD$:
\begin{align}
\rev(\bot) & \to \bot
 \tag{$\mathsf{r}_0$} \label{eq:rev-bot}\\
\rev(v(y)) & \to v(y)
 \tag{$\mathsf{r}_1$} \label{eq:rev-single}\\
\rev(\tie(y_1, y_2)) & \to \tie(\rev(y_2), \rev(y_1))
 \tag{$\mathsf{r}_2$} \label{eq:rev-2}
\end{align}

The conjecture is:
\begin{equation} 
\rev(\rev(x)) = x 
\label{eq:rev-conjecture}
\end{equation}
A proof of Conjecture~(\ref{eq:rev-conjecture}) can be found in~\cite{Kapur}.
We prove~(\ref{eq:rev-conjecture}) by the analysis 
of several cases, where each case is treated quickly.
As explained above, we need to decorate its variables with non-terminals of 
the normal form grammar.
There are three possibilities:
\begin{eqnarray}
 \rev(\rev(x)) & = & x \cons[ x \deriv \dxstate{x}{\List} ]
 	\label{ex:rev1}\\
 \rev(\rev(x)) & = & x \cons[ x \deriv \state{\bot} ]
 	\label{ex:bot}\\
 \rev(\rev(x)) & = & x \cons[ x \deriv \state{\tie(x_1, x_2)} ]
 	\label{ex:rev-tie}
\end{eqnarray}

Let us apply the production rules of the grammar to Conjectures~(\ref{ex:rev1})
and~(\ref{ex:bot}) (inference \textsf{Inductive Narrowing}). It returns respectively:
\begin{eqnarray}
 \rev(\rev(v(x))) & = & x \cons[ x \deriv \dxstate{x}{\Nat} ]
 	\label{ex:rev2-v}\\
 \rev(\rev(\bot)) & = & \bot 
 	\label{ex:rev1-bot}
\end{eqnarray}
The subgoals (\ref{ex:rev2-v}) and~(\ref{ex:rev1-bot}) 
are reduced by the rules 
(\ref{eq:rev-single}) and (\ref{eq:rev-bot}) of $\RD$ 
(\textsf{Inductive Rewriting} for \textsf{Inductive Narrowing})
into the respective tautologies:
$v(x) = v(x) \cons[ x \deriv \dxstate{x}{\Nat} ]$ and $\bot = \bot$.

Now, let us apply \textsf{Inductive Narrowing} to Conjecture~(\ref{ex:rev-tie}).
The application of the production rules of the grammar $\G_\NF(\RC)$ returns:
\begin{multline}
 \rev(\rev(\tie(x_1, x_2)))  =  \tie(x_1, x_2) \\
 \cons[ x_1 \deriv \dxstate{x_3}{\List} \wedge x_2 \deriv \dxstate{x_4}{\List} \wedge x_3^\List \sim x_4^\List ]
 \label{ex:rev-tie1}
\end{multline}
\begin{multline}
 \rev(\rev(\tie(x_1, x_2)))  =  \tie(x_1, x_2)  \\
{\cons[ x_1 \deriv \dxstate{x_3}{\List} \wedge x_2 \deriv \state{\tie(x_4, x_5)} \wedge x_3^\List \sim \tie(x_4, x_5) ]}
 \label{ex:rev-tie2}
\end{multline}
\begin{multline}
 \rev(\rev(\tie(x_1, x_2)))  =  \tie(x_1, x_2)  \\
 {\cons[ x_1 \deriv \state{\tie(x_3, x_4)} \wedge x_2 \deriv \dxstate{x_5}{\List} \wedge \tie(x_3, x_4) \sim x_5^\List ]}
 \label{ex:rev-tie3}
\end{multline}
\begin{multline}
 \rev(\rev(\tie(x_1, x_2)))  =  \tie(x_1, x_2)  \\
 {\cons[ x_1 \deriv \state{\tie(x_3, x_4)} \wedge x_2 \deriv \state{\tie(x_5, x_6)} \wedge \tie(x_3, x_4) \sim \tie(x_5, x_6) ]}
 \label{ex:rev-tie4}
\end{multline}
Note that, with~(\ref{eq:rev-2}):
\[
\begin{array}{rcl}
 \rev(\rev(\tie(x_1, x_2))) & \to_\RD & \rev(\tie(\rev(x_2), \rev(x_1))) \\
	& \to_\RD & \tie(\rev(\rev(x_1)), \rev(\rev(x_2))) 
\end{array}	
\] 
Hence, the reduction of~(\ref{ex:rev-tie1}) with the rule~(\ref{eq:rev-2}) of $\RD$ gives:
\begin{multline}
 \tie(\rev(\rev(x_1)), \rev(\rev(x_2)))   =  \tie(x_1, x_2) \\
 \cons[ x_1 \deriv \dxstate{x_3}{\List} \wedge x_2 \deriv \dxstate{x_4}{\List} \wedge x_3^\List \sim x_4^\List ]
 \label{eq:rev-subgoal1}
\end{multline}
ans similarly for ~(\ref{ex:rev-tie2}), ~(\ref{ex:rev-tie3}), and~(\ref{ex:rev-tie4}).

This later equation~(\ref{eq:rev-subgoal1}) 
can be reduced by Conjecture~(\ref{ex:rev-tie}), 
considered as an induction hypothesis (this is a case of \textsf{Inductive Rewriting}), 
giving the tautology:
\begin{equation}
 \tie(x_1, x_2))   =  \tie(x_1, x_2)
 \cons[ x_1 \deriv \dxstate{x_3}{\List} \wedge x_2 \deriv \dxstate{x_4}{\List} \wedge x_3^\List \sim x_4^\List ]
\end{equation}
The situation is the same for the other reduced equation and this completes
the proof of Conjecture~(\ref{ex:rev-tie}).

\end{RR}

\section{Conclusion}   
\begin{ABS}
We have proposed a procedure for automated inductive theorem proving
in specification made of conditional and constrained rewrite rules.
Constraints in rules can serve to 
transform non terminating specifications into terminating ones
(ordering constraints),
define ad-hoc evaluation strategies
(normal form constraints),
or for the analysis of trace properties of
infinite state systems like security protocols
(constraints of membership in a regular tree language representing
faulty traces~\cite{BouhoulaJacquemard07arspa}).
The expressiveness and efficiency of the procedure are obtained 
by the use of constrained tree grammars as a finite representation
of the initial model of specifications.
They are used both as induction schema and for decision procedures.
\end{ABS}

\begin{RR}
A fundamental issue in automatic theorem proving by induction 
is the computation of a suitable finite description of the set 
of ground terms in normal form, which can be used as an induction scheme.
Normal form constrained tree grammars are perfect induction schemes
in the sense that they generate \emph{exactly} the set of 
constructor terms in normal form. 
At the opposite, test sets and cover sets are approximated induction schemes
when the constructors are not free.
They may indeed also represent some reducible ground terms, 
and therefore may cause the failure (a result of the form ``don't know'') 
of an induction proof when constructors are not free.
In this case, refutational completeness is not guaranteed.
This explains the choice of constrained grammars for the incremental generation of subgoals. 
Constrained tree grammars are also used (by mean of emptiness test) 
in order to detect in some cases that constructor subgoals are inductively valid.
Moreover, this formalism permits to handle naturally constraint of 
membership in a fixed regular tree language.

Our inference system allows rewrite rules between constructors which can  be constrained.
Hence it permits to automate induction proofs on complex data structures.
It is sound and refutationally complete,
and allows for the refutation of false conjectures,
even with constrained constructor rules.
Moreover, all the conditions of inference rules are either recursive calls to
the procedure (\textsf{Rewrite~Splitting} or \textsf{Inductive~Contextual~Rewriting}),
or either some tests decidable 
under some assumptions on the constraints of the rewrite system for constructors.
These assumptions are required for decision of emptiness of constrained grammar languages.

Constraints in rules can serve to 
transform non terminating specifications into terminating ones,
for instance in presence of associativity and commutativity axioms
(ordering constraints),
define ad-hoc evaluation strategies, 
like e.g. innermost rewriting, directly in the axioms
(normal form constraints),
or for the analysis of trace properties of
infinite state systems like security protocols
(constraints of membership in a regular tree language representing
faulty traces~\cite{BouhoulaJacquemard07arspa}).
The treatment of membership constraints permits 
to express in a natural way, in conjectures, trace properties for the verification
of systems. This idea has been applied 
for the validation and research of attacks (by refutation)
on security protocols 
in a model with explicit destructor functions~\cite{BouhoulaJacquemard07arspa}.
These symbols represent operators like projection or decryption
whose behaviour is specified with constructor axioms.

%

Our procedure can handle partial specifications:
specifications which are not sufficiently complete
and specifications with partial constructor functions in the lines of~\cite{Kapur}.
Moreover, it preserves the proofs 
of decorated conjectures
made in a sufficiently complete specification
when this specification is extended with partial symbols.

\end{RR}

\begin{RR}
\label{se-extension}
\bigskip
The definition of tree grammars with constraints in Section~\ref{se-grammar} is very general. 
It embeds some classes of grammars 
for which the emptiness problem is decidable
(see Section~\ref{sb-scheme})
and also classes for which this problem is still open. 
Therefore, advances in tree automata theory can benefit our approach, 
and we are planing to study new classes of tree automata with constraints.



\end{RR}

\subsubsection*{Acknowledgments.} 
We wish to thank Michael Rusinowitch, 
Hubert Comon-Lundh, Laurent Fribourg and Deepak Kapur
for the fruitful discussions that we had together regarding this work.
We are also grateful to Jared Davis and Sorin Stratulat for having processed 
the example on sorted lists with 
respectively 
\textsf{ACL2} and \textsf{SPIKE}.



\begin{thebibliography}{10}

\begin{RR}
\bibitem{bogaert92stacs}
B.~Bogaert and S.~Tison.
\newblock Equality and disequality constraints on brother terms in tree automata.
\newblock In {\em Proc. of the 9th Symp. on Theoretical Aspects of Computer Science}, 1992.
\end{RR}

\begin{RR}
\end{RR}

\bibitem{Bouhoula97jsc}
A.~Bouhoula.
\newblock Automated theorem proving by test set induction.
\newblock {\em Journal of Symbolic Computation}, 23(1):47--77, 1997.

\bibitem{BouhoulaJacquemard07arspa}
A.~Bouhoula and F.~Jacquemard.
\newblock Verifying regular trace properties of security protocols with
  explicit destructors and implicit induction.
\newblock In {\em {P}roc. of the workshop {FCS-ARSPA}},
pages 27--44, 2007.

\bibitem{BouhoulaJacquemard08ijcar}
A.~Bouhoula and F.~Jacquemard.
\newblock Automated Induction with Constrained Tree Automata.
\newblock in {\em Proceedings of the 4th {I}nternational {J}oint {C}onference on {A}utomated {R}easoning (IJCAR)},
  vol. 5195 of Springer LNCS, pages 539--553, 2008.

\bibitem{BouhoulaJouannaud01}
A.~Bouhoula and J.-P. Jouannaud.
\newblock Automata-driven automated induction.
\newblock {\em Information and Computation}, 169(1):1--22, 2001.

\bibitem{BouhoulaJouannaudMeseguer}
A.~Bouhoula, J.-P. Jouannaud, and J.~Meseguer.
\newblock Specification and proof in membership equational logic.
\newblock {\em Theoretical Computer Science}, 236(1-2):35--132, 2000.

\bibitem{BouhoulaRusinowitch95jar}
A.~Bouhoula and M.~Rusinowitch.
\newblock Implicit induction in conditional theories.
\newblock {\em Journal of Automated Reasoning}, 14(2):189--235, 1995.

\begin{RR}
\bibitem{CaronComonCoquideDauchetJacquemard94}
A.C.~Caron, H.~Comon, J-L.~Coquid{\'e}, M.~Dauchet, and F.~Jacquemard.
\newblock Pumping, cleaning and symbolic constraints solving.
\newblock In {\em Proc. of the 21st  Int. Conf. on Automata, Languages and Programming}, 1994.
\end{RR}


\bibitem{ComonThese88}
H.~Comon.
\newblock {\em Unification et disunification. {T}h{\'e}ories et applications}.
\newblock PhD thesis, Institut Polytechnique de Grenoble (France), 1988.

\bibitem{tata}
H.~Comon, M.~Dauchet, R.~Gilleron, F.~Jacquemard, C.~L{\" o}ding, D.~Lugiez,
  S.~Tison, and M.~Tommasi.
\newblock Tree automata techniques and applications.
\newblock \url{http://www.grappa.univ-lille3.fr/tata}, 2007.

\bibitem{ComonJacquemard03}
H.~Comon and F.~Jacquemard.
\newblock Ground reducibility is exptime-complete.
\newblock {\em Information and Computation}, 187(1):123--153, 2003.

\bibitem{Comon01handbook}
H.~Comon-Lundh.
\newblock {\em Handbook of Automated Reasoning}, chapter Inductionless
  Induction.
\newblock Number chapter 14. Elsevier, 2001.

\bibitem{CJP-fossacs07}
H.~Comon{-}Lundh, F.~Jacquemard, and N.~Perrin.
\newblock Tree automata with memory, visibility and structural constraints.
\newblock In {\em {P}roc. of the 10th {I}nt. {C}onf. on {F}ound. of
  {S}oftware {S}cience and {C}omp. {S}truct. ({FoSSaCS}'07)}, 
  vol. 4423 of {\em LNCS}, pages 168--182. Springer, 2007.

\begin{RR}
\bibitem{dauchet95jsc}
M.~Dauchet, A.-C.~Caron, and J.-L.~Coquid{\'e}.
\newblock Automata for reduction properties solving.
\newblock {\em Journal of Symbolic Computation}, 20, 1995.
\end{RR}

\bibitem{Davis-osets}
J.~Davis.
\newblock Finite set theory based on fully ordered lists.
\newblock In {\em In 5th International Workshop on the ACL2 Theorem Prover and
  Its Applications (ACL2 2004)}, 2004.
\newblock Sets Library Website:
  \url{http://www.cs.utexas.edu/users/jared/osets/Web}.

\bibitem{DershowitzJouannaud90}
N.~Dershowitz and J.-P. Jouannaud.
\newblock Rewrite systems.
\newblock In {\em Handbook of Theoretical Computer Science, Volume B: Formal
  Models and Sematics}, pages 243--320. MIT Press, 1990.

\begin{RR}
\bibitem{JGL-ipl2005}
J.~Goubault{-}Larrecq.
\newblock Deciding {\(\mathcal{\MakeUppercase{H}}_1\)} by Resolution. 
\newblock {\em Information Processing Letters}, 95(3):401--408, 2005.
\end{RR}

\bibitem{JRV-jlap08}
F.~Jacquemard, M.~Rusinowitch, and L.~Vigneron.
\newblock Tree automata with equality constraints modulo equational theories.
\newblock {\em Journal of Logic and Algebraic Programming}, 75(2), pages 182--208, 2008.

\begin{RR}
\bibitem{JouannaudKounalisLICS86}
J.-P. Jouannaud and E.~Kounalis.
\newblock Proof by induction in equational theories without constructors.
\newblock In {\em Proc. 1st IEEE Symposium on Logic in Computer Science}, 1986.
\end{RR}

\bibitem{Kapur}
D.~Kapur.
\newblock Constructors can be partial too.
\newblock In {\em Essays in Honor of Larry Wos}. MIT Press, 1997.

\begin{RR}
\bibitem{Kapur91acta}
D.~Kapur, P.~Narendran, D.~Rosenkrantz, and H.~Zhang.
\newblock Sufficient completeness, ground reducibility and their complexity.
\newblock {\em Acta Informatica}, 28:311--350, 1991.
\end{RR}

\begin{RR}
\bibitem{ACL2}
M.~Kaufmann, P.~Manolios, and J.S.~Moore.
\newblock Computer-Aided Reasoning: An Approach.
\newblock Kluwer Academic Publishers, 2000.
\end{RR}

\bibitem{KirchnerKirchnerRusinowitch90ria}
C.~Kirchner, H.~Kirchner, and M.~Rusinowitch.
\newblock Deduction with symbolic constraints.
\newblock {\em Revue d'Intelligence Artificielle}, 4(3):9--52, 1990.
\newblock Special issue on Automatic Deduction.

\begin{RR}
\bibitem{Misra94powerlist}
Jayadev Misra.
\newblock Powerlist: {A} structure for parallel recursion.
\newblock {\em ACM Transactions on Programming Languages and Systems},
  16(6):1737--1767, 1994.
\end{RR}

\begin{RR}
\bibitem{Paulson98}
Lawrence~C. Paulson.
\newblock The inductive approach to verifying cryptographic protocol.
\newblock {\em Journal of Computer Security}, 6:85--128, 1998.
\end{RR}

\begin{RR}
\bibitem{plaisted85ic}
David~A. Plaisted.
\newblock Semantic confluence tests and completion methods.
\newblock {\em Information and Control}, 65(2-3):182--215, 1985.
\end{RR}

\begin{RR}
\bibitem{sengler96termination}
C.~Sengler.
\newblock Termination of Algorithms over Non-freely Generated Data Types.
\newblock In proceedings of the 13th Int. Conf. on Automated Deduction, 
   vol. 1104 of  Springer LNCS, pages 121-135, 1996.
\end{RR}

\bibitem{stratulat01}
S.~Stratulat.
\newblock A general framework to build contextual cover set induction provers.
\newblock {\em Journal of Symbolic Computation}, 32(4):403--445, 2001.

\bibitem{Zhang-CTRS92}
H.~Zhang.
\newblock Implementing contextual rewriting.
\newblock In {\em In Proc. 3rd Int. Workshop on Conditional Term Rewriting
  Systems}, 1992.
  
\begin{RR}
\bibitem{ZhangKKCADE88}
H.~Zhang, D.~Kapur, and M.~S. Krishnamoorthy.
\newblock A mechanizable induction principle for equational specifications.
\newblock In {\em Proc. 9th  Int. Conf. on Automated Deduction}, vol. 310 
  of Springer LNCS, pages 162--181, 1988.
\end{RR}
  

\end{thebibliography}

\end{document}